\documentclass[11pt]{llncs}
\usepackage{epsfig,amsfonts,latexsym,amssymb,amsmath}
\usepackage[T1]{fontenc}
\usepackage{amsmath,amssymb}
\usepackage{algorithm}
\usepackage{algorithmic}
\let\doendproof\endproof
\renewcommand\endproof{~\hfill\qed\doendproof}

\newcommand{\ignore}[1]{}

\def\srl{\textsc{SRL}}
\def\srd{\textsc{SRD}}

\def\srda{\textsc{SRD}$_\alpha$}
\def\srdr{\textsc{SRD}$^\rho$}

\def\reals{\mathbb{R}}

\newcommand{\sinn}[1]{\sin \left({#1}\right)}
\newcommand{\coss}[1]{\cos \left({#1}\right)}
\newcommand{\cscc}[1]{\csc \left({#1}\right)}

\newcommand{\arccoss}[1]{\arccos \left({#1}\right)}
\newcommand{\arcsinn}[1]{\arcsin \left({#1}\right)}

\newcommand{\ren}[4]{\mathcal{R}_{#1}^{#2}\left(#3,#4 \right)}
\newcommand{\ene}[4]{\mathcal{E}_{#1}^{#2}\left(#3,#4 \right)}

\begin{document}


\title{Symmetric Rendezvous With Advice:\\
How to Rendezvous in a Disk
\thanks{This is the full version of the paper with the same title which will appear in the proceedings of the 
25th International Colloquium on Structural Information and Communication Complexity, June 18-21,  2018, Ma'ale HaHamisha, Israel.}
}

\author{
Konstantinos Georgiou\inst{1}
\thanks{Research supported in part by NSERC Discovery Grant.}
\and
Jay Griffiths\inst{1}
\thanks{Research supported in part by NSERC Undergraduate Student Research Award.}
\and
Yuval Yakubov\inst{1}
\thanks{Research supported in part by the FoS Undergraduate Research Program, Ryerson University.}
}

\institute{
Department of Mathematics, Ryerson University \\
350 Victoria St, Toronto, ON, M5B 2K3, Canada \\
\email{konstantinos,jay.griffiths,yyakubov@ryerson.ca}
}

\maketitle

\begin{abstract}

In the classic Symmetric Rendezvous problem on a Line (\srl), two robots at known distance 2 but unknown direction execute the same randomized algorithm trying to minimize the expected rendezvous time. A long standing conjecture is that the best possible rendezvous time is 4.25 with known upper and lower bounds being very close to that value. 
We introduce and study a geometric variation of \srl\  that we call Symmetric Rendezvous in a Disk (\srd) where two robots at distance 2 have a common reference point at distance $\rho$. We show that even when $\rho$ is not too small, the two robots can meet in expected time that is less than $4.25$. 
Part of our contribution is that we demonstrate how to adjust known, even simple and provably non-optimal, algorithms for \srl, effectively improving their performance in the presence of a reference point.
Special to our algorithms for \srd\ is that, unlike in \srl, for every fixed $\rho$ the worst case distance traveled, i.e. energy that is used, in our algorithms is finite. In particular, we show that the energy of our algorithms is $O\left(\rho^2\right)$, 
while we also explore time-energy tradeoffs, concluding that one may be efficient both with respect to time and energy, with only a minor compromise on the optimal termination time.  
\end{abstract}

\section{Introduction}
In a rendezvous game two players reside at unknown locations in a given domain and they wish to minimize the (expected) meeting (rendezvous) time. Various rendezvous problems have been studied intensively, with applications in computer science and real-world modeling, such as the search for a mate problem in which species with a low spatial density try to find suitable partners~\cite{searchGamesBook}.
Rendezvous problems can be classified as \textit{asymmetric}, in which each agent may use a different strategy, or \textit{symmetric}, in which each agent follows the same algorithm; moreover, strategies can be classified as \textit{mixed},  incorporating randomness, or \textit{pure} which are deterministic.

In this paper, we discuss \textit{symmetric rendezvous with advice}. Two speed-1 robots (mobile agents) start at known distance but at unknown locations and they are trying to meet (rendezvous). At any time, robots have the option to meet at a known immobile reference point that is initially placed $\rho$ away from both agents. The goal is to design mixed strategies so as to minimize the expected rendezvous time, i.e. the expected value of the first time that robots meet. 
After scaling, our problem can be equivalently described as a Symmetric Rendezvous problem in a unit Disk (\srd), where mobile agents lie at  the perimeter of disk at known arc distance $2\alpha$, having the option to always meet at the origin. 

\srd\ is a geometric variation of the well-studied Symmetric Rendezvous problem on a Line (\srl) where no reference point is available, and for which a long-standing conjecture stipulates that it can be solved in expected time 4.25. Critical differences between the two problems is that in \srd\ 
(a) the rendezvous can always be realized deterministically,
(b) the performance can be much better than the distance from the reference point $\rho$ and better than the conjectured 4.25 even for not too small values of $\rho$ and
(c) the worst case rendezvous time can be bounded in $\rho$ even when one tries to minimize the expected rendezvous time. 
The latter is an important property, since if the two agents are vehicles with limited fuel, our strategies can be used to guarantee rendezvous before the fuel runs out.

\subsection{Related Work}

The rendezvous problem is a special type of a search game where two or more agents (robots) attempt to occupy the same location at the same time in a domain. 
Search games and rendezvous have a long history; see~\cite{searchGamesBook} and~\cite{Alpern2013} for a thorough introduction to the area, and~\cite{Alpern2002} for a not so recent survey.
The challenge of the task (search or rendezvous) is induced by limitations related to communication, coordination, synchronization, mobility, visibility, or other types of resources, whereas examples of rendezvous domains include networks, discrete nodes and geometric environments. Notably, each of the aforementioned specifications, along with combinations of them, have given rise to a long list of publications, a short representative list of which we discuss below. 

The rendezvous problem was first proposed informally by Alpern~\cite{Alpern76} in 1976, and received attention due to the seminal works of Anderson and Weber~\cite{discreteLocations} for discrete domains and of Alpern~\cite{Alpern1995} for continuous domains. Our work is a direct generalization of the special and so-called Symmetric Rendezvous Search Problem on a Line (\srl) proposed by Alpern~\cite{Alpern1995} in 1995. In that problem, two blind agents are at known distance 2 on a line, and they can perform the same synchronized randomized algorithm (with no shared randomness). The original algorithm of Alpern~\cite{Alpern1995} had performance (expected rendezvous time) 5, which was later improved to 4.5678~\cite{lineWeirdPlayers}, then to 4.4182~\cite{Baston99}, then to 4.3931~\cite{uthaisombut2006symmetric}, and finally to the best performance known of 4.2574~\cite{hanetal} by Han et al. Similarly, a series of proven lower bounds~\cite{Alpern1995RSL},~\cite{uthaisombut2006symmetric} have lead to the currently best value known of 4.1520~\cite{hanetal}. 

A number of variations of \srl\ have been exhaustively studied, and below we mention just a few. 
The symmetric rendezvous problem with unknown initial distance or with partial information about it has been considered in \cite{beveridge2011symmetric} and 
\cite{baston1998rendezvous}. 
A number of different topologies have been considered including 
labeled network~\cite{alpern2002rendezvous},
labeled line~\cite{CT04}, 
ring~\cite{kranakis2003mobile},~\cite{flocchini2004multiple} (see survey monograph~\cite{KKM10}), 
torus~\cite{kranakis2006mobile}, 
planar lattice~\cite{alpern2005rendezvous}, and 
high dimensional host spaces~\cite{alpern2006rendezvous}.
We note here that the topology we consider in this work follows a long list studies of relevant search/rendezvous-type problems in the disk. 
The rendezvous problem with faulty components has been studied in~\cite{Das08c} and~\cite{DLM15}.
Asynchronous strategies have been explored in~\cite{TZ14} and~\cite{Prencipe07}.
Studied variations of robots capabilities include 
sense of direction~\cite{alpern2006common},~\cite{BFFS07},
memory~\cite{CFR09},
visibility~\cite{CPL12}, 
speed~\cite{feinerman2014fast}, 
power consumption~\cite{ACCLPV12}
and location awareness~\cite{collins2011synchronous}.
Interesting variations of communication models between agents have been studied in 
\cite{Das2007} (whiteboards),
\cite{czyzowicz2008power} (tokens),
\cite{KKM10} (mobile tokens), and 
\cite{Prencipe07} (look-compute-move model). 
Finally, \cite{pelc2012deterministic} is a comprehensive survey in deterministic rendezvous in networks.

\subsection{Formal Definitions, Notation \& Terminology}
\label{sec: definitions}

\subsubsection{Problem Definition}
In the Symmetric Rendezvous problem in a Disk (\srd) two agents (robots) are initially placed on the plane at known distance from each other but at unknown location. A common reference point $O$ is at known distance and known location to both robots. The robots can move at speed 1 anywhere on the plane, and they detect each other only if they are at the same location, i.e. when the meet. Given that robots run the same (randomized) and synchronized algorithm, the goal is to design trajectory movements so as to minimize the (expected) meeting, also known rendezvous, time. 

The natural way to model \srd\ is to have robots start on the perimeter of disk, where its center serves as the common reference point. We adopt two equivalent parameterizations of the problem that arise by either normalizing robots' initial distance or the radius of the disk. In \srdr\ the disk has radius $\rho$, and the robots have Euclidean distance 2, while in \srda\ robots start on the perimeter of a unit disk and their arc distance is $2\alpha$. 

As we explain below, \srdr\ is the natural extension of the well-studied rendezvous on a line problem, while \srda\ is convenient for analyzing the performance of trajectory movements. We will use both perspectives of the problem interchangeably. 
Clearly, the initial Euclidean distance of the two robots in \srda\ is $2\sinn{\alpha}$. Hence, after scaling the instance by $1/\sinn{\alpha}$, the initial distance of the robots becomes 2, and the reference point (the origin) is at distance $\rho=1/\sinn{\alpha}$. Therefore, \srdr\ and \srda\ are equivalent under transformation $\alpha=\arcsinn{1/\rho}$.
Moreover, we will silently assume that $0<\alpha<\pi/4$ as otherwise \srda\ is degenerate, or that $\rho>\sqrt{2}$ for \srdr.

\subsubsection{The Related Rendezvous on a Line Problem}
In the well-studied Rendezvous problem on a Line (\srl), two robots, with the same specifications as in \srd\ are placed at known distance 2, but at unknown locations on the line. The objective is again to minimize the (expected) rendezvous time. Note that \srl\ is exactly the same as \srd$^\infty$. 

Natural randomized algorithms for solving \srl\ are so-called $k$-\textit{Markovian Strategies}, i.e. random processes that iterate indefinitely, so that in every iteration each robot follows a partial trajectory of total length $k$ (or $k$ times more than the original distance of the agents). 
The simplest $2$-\textit{Markovian Strategy} achieves expected rendezvous time 7: each robot with probability 1/2 moves distance 1 to the left and then to the right, back to its original position (and robot follows the symmetric trajectory to the right with the complementary probability). Note that robots meet with probability 1/4 after time 1, and otherwise they repeat the experiment after moving distance 2. If $f$ denotes the expected meeting time, then clearly $ f = \frac14 + \frac34 (2+f)$ from which we obtain $f=7$. 

An elegant refinement was proposed by Alpern~\cite{Alpern1995} and achieves expected rendezvous time 5. In this $3$-\textit{Markovian Strategy} 
each robot with probability 1/2 moves distance 1 to the left, then to the right back to its original position and then further right at distance (and robot follows the symmetric trajectory to the right with the complementary probability). This time, robots meet with probability 1/4 after time 1, and with probability 1/4 after time 3, otherwise the repeat the same process. If $f$ denotes the expected meeting time, then 
$ f = \frac14 + \frac143+\frac12\left( 3+f\right)$ 
from which we obtain $f=5$. 
Interestingly, this is also the best possible 3-Markovian strategy. 

Alpern's algorithm above is a distance-preserving algorithm, that is, after each iteration robots either meet or they preserve their original distance (but not their original locations). After a series of improvements, this idea was fruitfully generalized to $k$-\textit{Markovian Strategies} by Han et al.~\cite{hanetal} giving the best known rendezvous time $4.2574$ (for $k=15$). Notably, the best lower bound know is $4.1520$~\cite{hanetal}, which has resulted into the believable conjecture that $4.25$ is the best rendezvous time possible.

\subsubsection{Measures of Efficiency}
\srd\ and \srl\ can be viewed as online problems, where robots attempt to solve the problem only with partial input information. The natural measure of efficiency of any proposed \textit{online algorithm} is the so-called competitive ratio, defined as the ratio between the (expected) online algorithm performance over the best possible performance achievable by an \textit{offline} algorithm that knows the input. With this terminology in mind, it is immediate that Alpern's Algorithm~\cite{Alpern1995} for \srl\ is $5$-competitive, while the conjecture above stipulates that 4.25 is the best possible competitive ratio for the problem. 

Using the terminology above, the best offline algorithm can solve \srdr\ in time 1, and \srda\ in time $\sinn{\alpha}$, hence for our competitive analysis we will always scale the expected performance of our randomized algorithms accordingly. As a result, the competitive ratio of our algorithms will be described by functions of $\rho$ and $\alpha$ for \srdr\ and \srda, respectively, that are at least 1 for all values of the parameters. 

Our main goal will be to beat the psychological threshold of $4.25$ for \srdr, even for not too small values of $\rho$, demonstrating this way both the usefulness of a reference point and the effectiveness of our algorithms. In order to quantify this more explicitly, we introduce one more alternative measure of efficiency: an algorithm for \srdr\ will be called $\delta$-\textit{effective}, if $\delta$ is the largest value of $\rho$ for which the expected rendezvous time is no more than 4.25. If such $\rho$ does not exist, i.e. if the algorithm has expected rendezvous time at least 4.25 for all $\rho>\sqrt{2}$, then we call the algorithm $0$-effective. 
%
%
To conclude, apart from calculating the competitive ratio of our algorithms for \srdr, we will complementarily comment also on the effectiveness, with the understanding that the the higher their value is, the better the algorithm is. 
Note for example that the naive algorithm that simply has robots go to the reference point is 
$\rho$-competitive and $4.25$-effective.

Finally, we also consider the worst case performance of our algorithms that we call \textit{energy}. Formally, the energy of a rendezvous algorithm is defined as the supremum of the time by when the rendezvous is realized with probability 1. Note that any algorithm for \srl\ is bound to have infinite energy, whereas we show in this paper a family of algorithms for \srd\ that have bounded energy.

\subsection{Our Results}

\subsubsection{Techniques Outline}
Our main contribution is the exploration of 3-Markovian strategies for \srd. In particular, we  adjust Alpern's optimal 3-Markovian algorithm~\cite{Alpern1995} so as to take advantage of the reference point. 
Similar to the algorithm for \srl, our algorithm uses infinitely many random bits. In each random step, robots attempt to meet twice. If the rendezvous is not realized, then the projection of their trajectory to the perimeter of the original disk has length 3, however agents reside in a smaller disk but still at the same arc-distance. Then, robots repeat the process, so that, overall, the distances of the possible meeting points to the origin are strictly decreasing, i.e. the disk is sequentially shrinking. The trajectories of the robots are determined by \textit{two critical angles}, that determine the distance of the possible meeting points to the origin, i.e. how much the disk are shrunk. 

If in each iteration, the disk is shrunk ``a lot'', then robots move much more than half their Euclidean distance in order to meet, however when they repeat the experiment, they are solving a simpler problem since they are at the same arc-distance but the reference point is closer. If, on the other hand, the new disk is comparable to the original one, then robots attempt to greedily rendezvous as fast as possible, however if the meeting is not realized, robots have to solve an identical rendezvous problem (and such a strategy is bound to have a competitive ratio no better than 5, i.e. the ratio of the original \srl). 
Hence, the heart of the difficulty is to determine the two critical angles so that the instance that robots have to solve in each step shrinks by the right amount. 
Part of our contribution is that we demonstrate how to model the latter problem as a non-trivial non-linear optimization problem, which we also solve. 

\subsubsection{High Level Contributions}
As it is typical in online algorithmic problems, the impossibility of achieving optimal solutions is due to the unknown input (in our case the exact location of the robots). Our work contributes toward the fundamental algorithmic question as to whether additional resources (partial information about the unknown input - in our case a reference point) could yield improved upper bounds. Not only we answer this question in the positive, and we quantify properly our findings, but our trajectories also demonstrate how a rendezvous can be realized in 2 dimensions, even though the detection visibility of the robots in one dimensional. 
Part of our contribution is to also demonstrate how to adjust known algorithms for \srl\ so as to solve \srd. 
In particular, our methods can be generalized and induce improved competitive ratio upper bounds when the starting rendezvous algorithm is some other $k$-Markovian trajectory, $k>3$ (see~\cite{hanetal}). 
However, each such adaptation requires the determination of more than two critical angles, and the induced non-linear optimization problems would be possible to solve only numerically, rather than analytically as we do in this work. 
At the end, our algorithms are simple, yet powerful enough to induce good performance for a wide range of \srd\ instances. 

\subsubsection{Discussion on Energy}
We also consider the \textit{worst case} rendezvous time for our algorithms that we deliberately call \textit{energy}. In real-life applications, robots are bound to run only for limited time due to restricted resources (e.g. fuel). Assuming that the actual energy spent (fuel burnt) by a robot is proportional to it's operation time, we view the worst-case running time of our algorithms as the minimum \textit{energy} required by the robots that ensures that the execution of the algorithm terminates successfully with probability 1. 
Note that in the original \srl\ problem, and for any feasible rendezvous strategy, there is a positive probability (though exponentially small) that the rendezvous is arbitrarily large. Given that mobile robots should have access to bounded \textit{energy} (fuel), 
the probability that the rendezvous is never realized is positive. In contrast, we show that our algorithms for \srd\ require bounded energy, that there is a finite time by when the rendezvous is realized with probability 1. We show that this property holds true under mild conditions for our algorithms, and in particular it holds true for our algorithm that minimizes the expected rendezvous time. 
For the latter algorithm we show that the energy required in $\Theta\left(\rho^2\right)$. Finally, and somehow surprising, we also show that by compromising slightly on the expected termination time, the required energy becomes $\Theta\left(\rho\right)$.

\subsubsection{Paper Organization}

Section ~\ref{sec: upper bounds} is devoted to the optimization problem of minimizing the expected rendezvous time. 
First, in Section~\ref{sec: benchmark} we introduce some simple rendezvous algorithms that are mostly used as benchmark results for what will follow. 
Section~\ref{sec: 1 random bit} introduces the first non-trivial refinement, by providing a single random bit 1-Markovian algorithm. Our observations and results of that section are later used in Section~\ref{sec: 3-markovian}, where we discuss general $3$-Markovian strategies. Our main contribution is the determination of optimal critical angles, as well as of the induced competitive ratio, and induced effectiveness. We also provide the asymptotic behavior of the critical angles, as well as the convergence to competitive ratio 5, as the distance $\rho$ of the reference point goes to infinity. 
Then, in Section ~\ref{sec: energy} we study the worst case rendezvous time induced by our most efficient algorithm for \srd.  
In particular, 
the main contribution of Section~\ref{sec: energy 3-markovian infty} is the asymptotic analysis of the worst case rendezvous time for our algorithm that is meant to minimize the expected rendezvous time, and is shown to be $\Theta\left(\rho^2\right)$. 
Motivated by this, we study in Section~\ref{sec: constrained} time-energy tradeoffs. More specifically, we show that asymptotically in $\rho$, the expected termination time can stay optimal achieving improved but still $\Theta\left(\rho^2\right)$ energy, while only slightly suboptimal termination time allows for $\Theta\left(\rho\right)$ energy. 
Our expected rendezvous time positive results for \srd\ are summarized in Figure~\ref{fig: comparison}.
\begin{figure}[h!]
\begin{center}
 \includegraphics[width=8cm]{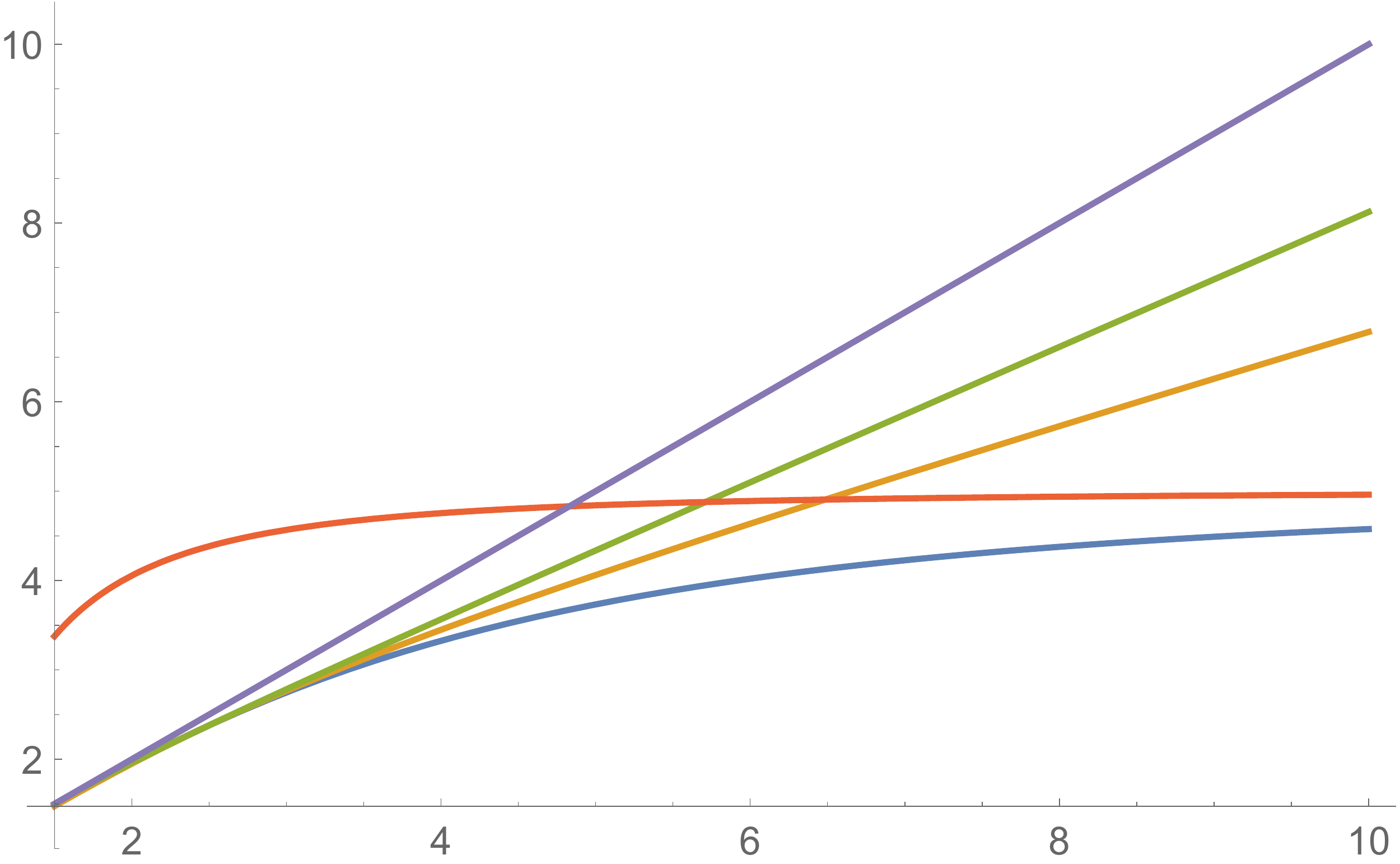}
\end{center}
 \caption{A comparison between the competitive ratio of the discussed algorithms for \srdr. 
 The horizontal axis corresponds to $\rho$, and the vertical to the competitive ratio. 
 The curves, along with the corresponding theorems that establish each result are as follows: 
 purple curve is the naive ``go-to-origin'' $\rho$-competitive 4.25-effective algorithm,
 green curve is the 4.888-effective 1-random bit Algorithm due to Theorem~\ref{thm: opt 1rbbeta},
 yellow curve is the 5.3236-effective Algorithm due to Theorem~\ref{thm: opt 1rbbetagamma},
 red curve is the 2.57-effective Algorithm due to Theorem~\ref{thm: 3-mark easy},
 and blue curve is the 7.1367-effective Algorithm due to Theorem~\ref{thm: opt parameters for inftyRB}. 
 }
 \label{fig: comparison}
\end{figure}
Many of our calculations throughout the paper are assisted by computer symbolic software (\textsc{Mathematica}), but all our results are rigorous.
Appendix~\ref{appendix: proofs} contains additional technical lemmata (and their proofs) omitted altogether from the main body, and which are invoked throughout this paper.

\ignore{
\begin{figure}[h!]
\begin{center}
 \includegraphics[width=7cm]{figs/comparison}
\end{center}
 \caption{A comparison between the competitive ratio of the discussed algorithms for \srdr. 
 The horizontal axis corresponds to $\rho$, and the vertical to the competitive ratio. 
 The curves, along with the corresponding theorems that establish each result are as follows: 
 purple curve is the naive ``go-to-origin'' $\rho$-competitive 4.25-effective algorithm,
 green curve is the 4.888-effective 1-random bit Algorithm due to Theorem~\ref{thm: opt 1rbbeta},
 yellow curve is the 5.3236-effective Algorithm due to Theorem~\ref{thm: opt 1rbbetagamma},
 red curve is the 2.57-effective Algorithm due to Theorem~\ref{thm: 3-mark easy},
 and blue curve is the 7.1367-effective Algorithm due to Theorem~\ref{thm: opt parameters for inftyRB}. 
 }
 \label{fig: comparison}
\end{figure}
}

\ignore{
naive[\[Rho]_] := (7*\[Rho]^2 + 8*Sqrt[\[Rho]^2 - 1]*\[Rho] - 
     3)/(3*\[Rho]^2 + 1);
OneRBbeta[\[Rho]_] := 1/4 (Sqrt[7] + 3 Sqrt[-1 + \[Rho]^2]) (*Thm 4*)

OneRBbetagamma[\[Rho]_] := 
 3/4 Sqrt[1 - 
    1/\[Rho]^2] (2/3 \[Rho] - 4/(3 \[Rho]) + (
     2 Sqrt[5] Sqrt[1 - 1/\[Rho]^2])/3) + Sqrt[
  1 - 9/16 (2/3 - 4/(3 \[Rho]^2) + (2 Sqrt[5] Sqrt[1 - 1/\[Rho]^2])/(
      3 \[Rho]))^2]; (*Thm 6*)
InfRBbetagamma[\[Rho]_] := \[Rho]*
  exprendtime[ArcSin[1/\[Rho]]]; (*Thm 7*)

Plot[ {InfRBbetagamma[\[Rho]], OneRBbetagamma[\[Rho]], 
  OneRBbeta[\[Rho]], naive[\[Rho]], \[Rho]}, {\[Rho], 1.5, 10}]
  }

\section{Rendezvous Algorithms in a Disk}
\label{sec: upper bounds}

\subsection{Some Immediate Benchmark Upper Bounds}
\label{sec: benchmark}

First we establish some immediate positive results that can be used as benchmarks for  rendezvous trajectories that we will present in subsequent sections. Recall that the naive ``go-to-origin'' algorithm is 4.25-effective. 

The first attempt is to blindly implement the 4.2574-competitive algorithm of~\cite{hanetal} for \srl. Indeed, given instance \srda, robots can be restricted to move on the perimeter of the disk. It is clear that the resulting algorithm has expected rendezvous time $\alpha$, and hence competitive ratio $4.2574 \frac{\alpha}{\sinn{\alpha}}$ for \srda\ (note that $\frac{\alpha}{\sinn{\alpha}}\geq 1$). 
However, one can slightly improve upon this by making robots move along chords instead. Indeed, the algorithm of \cite{hanetal} for \srl\ has the property that robots always move and attempt to meet at integral points, assuming that one of the robots starts from the origin of the real line. Now for problem \srdr\ in the disk, and given any initial location of the robots, consider an infinite sequence of clockwise and of counterclockwise arcs of length 2, along with their corresponding chords of length $2\sinn{1}$. Any integral movement of robots in the line can be simulated by movements on the chords by multiples of $\sinn{1}$, while $\sinn{1}$ is also the optimal offline solution. 
Therefore, we immediately obtain the following. 
\begin{theorem}\label{lem: Han}
\srdr\ admits an online algorithm which is $4.2574$-competitive and $0$-effective. 
\end{theorem}

Next we show that Theorem~\ref{lem: Han} admits an easy refinement using a simple 3-Markovian process, which is a direct application of~\cite{Alpern1995}.

\begin{theorem}
\label{thm: 3-mark easy}
\srdr\ admits an online algorithm which is 
$\left(\frac{7 \rho ^2+8 \sqrt{\rho ^2-1} \rho -3}{3 \rho ^2+1}\right)$-competitive
and $2.57$-effective.
\end{theorem}

\begin{proof}
We introduce the language of \srda. The main idea of the algorithm is that each robot iteratively tries to greedily meet her peer in the ``middle point of their locations''. More specifically, in each iteration, each robot tosses a coin, which advices the robot whether her peer is at arc distance $\alpha$ cw or ccw. Call the current robot's location $A$, say on a unit disk, and let $B$ be a point at cw distance $\alpha$. Then the robot attempts to meet her peer in the middle point $M$ of $A,B$, and this succeeds with probability 1/4. If this fails, it might be due to that the other robot was actually in the opposite direction and her random ccw move brought her at the corresponding point $M'$. Then the two robots attempt to meet in the middle point of $M,M'$, and again this meeting is realized with probability 1/4. In the complementary event, with probability 1/2 both robots choose to move in the same direction in their first move. Still after their second move, and given they have not met, they are still at arc-distance $\alpha$, but now they reside on a smaller disk, and they repeat the process. 

Denote  by $\mathcal R$ the expected rendezvous time of the algorithm, given that agents start on the perimeter of a radius-$1$ disk. 
If robots do not meet, they are still at arc distance $\alpha$ in a disk that is scaled by $\coss{\alpha}$. Therefore, their Euclidean distance in the resulting disk is $2\coss{\alpha} \sinn{\alpha} = \sinn{2\alpha}$.

Notice that with probability $1/4$ robots meet at time $\sinn{\alpha}$. With probability $1/4$, they meet at time 
$ \sinn{\alpha}+\tfrac12\sinn{2\alpha}$. Otherwise, the have already walked distance $\sinn{\alpha}+\tfrac12\sinn{2\alpha}$, and they are at arc distance $\alpha$ of a radius-$\coss{\alpha}$ disk, when they repeat the process.  Therefore, 
$$
\mathcal R =
\sinn{\alpha} + 
\frac34
\sinn{2\alpha} + \frac12\coss{\alpha} \mathcal R. 
$$
Solving for $\mathcal R$ gives  expected rendezvous time
$\mathcal R=\tfrac{\sin (a) (2+3 \cos (a))}{2-\cos (a)}$. 
Hence, the competitive ratio for \srda\ is $\tfrac{2+3 \cos (a)}{2-\cos (a)}$ and for \srdr\ it is $\tfrac{7 \rho ^2+8 \sqrt{\rho ^2-1} \rho -3}{3 \rho ^2+1}$. Note that the competitive ratio becomes 4.25 exactly for $\rho=\frac{29}{\sqrt{165}}\approx2.25765$. 
\end{proof}

\ignore{
FullSimplify[
 (Sin[a] + 3/4*Sin[2*a])*Sum[ (Cos[a]/2)^i, {i, 0, Infinity}]
 ]
test[a_] := -((2 + 3 Cos[a]) /(-2 + Cos[a]))
FullSimplify[test[ArcSin[1/\[Rho]]], Assumptions -> \[Rho] >= 2]
}

\subsection{Rendezvous with Minimal Randomness}
\label{sec: 1 random bit}
Theorems~\ref{lem: Han} and~\ref{thm: 3-mark easy} were obtained by algorithms that use infinitely many random bits. 
This section is devoted into showing that even with 1 random bit, we can perform better than the naive ``go-to-origin'' algorithm,
as well as of the algorithms of  Theorems~\ref{lem: Han} and~\ref{thm: 3-mark easy}, at least for certain values of $\alpha, \rho$. 
This will also help as a warm-up for our later results. 

Consider instance \srda\ and mobile agents at arc distance $\alpha$ as in Figure~\ref{1stepfig}. Each of them knows that their peer is $\alpha$ away either clockwise or counterclockwise, and consider the corresponding arcs. Notice that in both algorithms of Theorems~\ref{lem: Han} and ~\ref{thm: 3-mark easy} robots attempt to meet at the bisectors of the two arcs. Given a fixed angle $\beta$, each robot, and at each iteration chooses uniformly at random either the cw or the ccw direction, and moves in that direction with respect to the origin till the bisector is hit. We call this move a \textit{random $\beta$-darting}. Notice that $0$-darting corresponds to going to the origin, while the algorithm of Theorem~\ref{thm: 3-mark easy} we have $\beta=\pi/2-\alpha$. The main idea behind our 1-random bit algorithm 1RB with parameter $\beta$ is to choose the optimal $\beta\in [0,\pi/2-\alpha)$ that minimizes the expected termination time. 
\begin{figure}[h!]
\begin{center}
 \includegraphics[width=8cm]{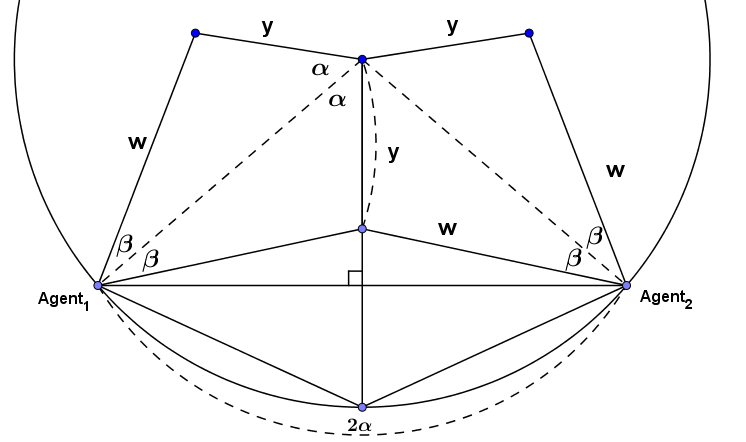}
 \end{center}
 \caption{Geometry of Algorithm 1RB$_\beta$.}\label{1stepfig}
\end{figure}
\begin{algorithm}
\caption{1RB$_\beta$}
\label{1RB}
\begin{algorithmic}[1]
\STATE Do a random $\beta$-darting.
\STATE Go to origin (if peer is not already met).
\end{algorithmic}
\end{algorithm}

\begin{lemma}
\label{lem: 1rb performance}
The expected rendezvous time $R(\beta)$ of 1RB$_\beta$ is 
$$
\mathcal R(\beta) = \sinn\alpha\csc(\alpha+\beta) + \frac34\sinn\beta\csc(\alpha+\beta).
$$
\end{lemma}

\begin{proof}
For fixed $\alpha$, let $w=w(\beta)$ be the length of the line segment between the position of a robot and the possible meeting point at the bisector of the critical arcs. Let also $y=y(\beta)$ denote the distance of the possible meeting point from the origin (see also Figure~\ref{1stepfig}). 

Clearly, with probability 1/4 robots move after time $w$, and otherwise they meet at time $y+w$. Hence
$$
\mathcal R(\beta) = \frac14y+\frac34(y+w) = y+\frac34w. 
$$
The proof follows by noticing that 
$w = \sinn\alpha\csc(\alpha+\beta)$ and that $y= \sinn\beta\csc(\alpha+\beta)$,
which is obtained by a simple geometric argument based on the Law of sines. 
\end{proof}

\begin{theorem}
\label{thm: opt 1rbbeta}
The optimal 1RB$_\beta$ algorithm uses 
$$
\overline{\beta}=\max \left(0,-\sinn{1/\rho} + \arccos\left(\tfrac{3}{4}\right)\right)
$$
in which case the algorithm is 
$\frac{3 \sqrt{\rho ^2-1}+\sqrt{7}}{4}$-competitive 
and 4.88813-effective. 
\end{theorem}

\ignore{
FullSimplify[
 \[Rho]*TrigExpand[
   Cos[ ArcCos[3/4] - ArcSin[1/\[Rho]]]
   ], Assumptions -> \[Rho] >= 1]
   }

\begin{proof}
For convenience, we analyze the performance on \srda\ instead. Using Lemma~\ref{lem: 1rb performance}, we find the critical values of the expected rendezvous time $\mathcal R(\beta)$ by calculating  
$$
\frac{d}{d\beta} \mathcal R(\beta)= \left(\tfrac{3}{4}-\cos\left(\alpha+\beta\right)\right)\csc^2\left(\alpha+\beta\right)\sin\left(\alpha\right). 
$$

Observe that as $0<\alpha\leq\alpha+\beta\leq\frac{\pi}{2}$ we have $\cos(\alpha+\beta)\leq \cos\alpha$ and thus, for $\alpha > \arccos{\frac{3}{4}}$ we see that $R(\beta)$ is increasing. Hence, $\mathcal R(\beta)$ is minimized at $\mathcal R(0)= \sin\alpha\csc\alpha = 1$.

For $\alpha \leq \arccos{\frac{3}{4}},$ $\mathcal R(\beta)$ is decreasing when $\beta< \overline{\beta}$ and increasing when $\beta > \overline{\beta}$, where
$\overline{\beta} = -\alpha + \arccos\left(\tfrac{3}{4}\right).$
Thus, $R(\beta)$ is minimized at $\overline{\beta}$ and by some straightforward trigonometric calculations 
we see that $\mathcal R(\overline{\beta})=\coss{\arccoss{\frac{3}{4}}-\alpha}$. Since $\arccos{\frac{3}{4}}-\alpha \in [0,\frac{\pi}{2}],$ then this cost is no more than 1.

When the problem is not degenerate, we conclude that 1RB$_{\overline{\beta}}$ is $\frac{\coss{\arccoss{\frac{3}{4}}-\alpha}}{\sinn{\alpha}}$-competitive. Our claim now follows for \srdr\ using transformation $\alpha=\sinn{1/\rho}$. Finally note that $\frac{3 \sqrt{\rho ^2-1}+\sqrt{7}}{4}$ is increasing, and it is equal to 4.25 when $\rho =\frac{1}{3} \sqrt{305-34 \sqrt{7}}\approx 4.88813.$
\ignore{
Reduce[ 1/4 (Sqrt[7] + 3 Sqrt[-1 + \[Rho]^2]) == 425/100, \[Rho]]
}
\end{proof}

\subsection{Improved Rendezvous with 3-Markovian Trajectories}
\label{sec: 3-markovian}

In this section we generalize the algorithm of Section~\ref{sec: upper bounds} in two ways; first we allow more random bits, and second, in every random trial, we allow robots trajectories two darting attempts (recall that Algorithm 1RB$_\beta$ allows for only one darting attempts. In the language of the established results for \srl\, we will adopt Alpern's 3-Markovian trajectory~\cite{Alpern1995}. 

The main idea behind our new algorithms is as follows

 At every random step, robots will reside at the perimeter of a disk, and they will be at constant arc distance $\alpha$. 
 As in 1RB$_\beta$, each robot is associated with two bisectors in which robot will make an attempt to meet her peer. 
 A fixed angle $\beta$ along with a random bit will determine the direction (cw or ccw) of the random $\beta$-darting that will bring the robot in one of the bisectors. 
 Note that due to the symmetry imposed by the trajectory, a meeting is realized in this step with probability 1/4. 
If the rendezvous is not realized, the robot will attempt a deterministic $\gamma$-darting to the other bisector, and the meeting is realized in this step with probability 1/4 as well. 
If the rendezvous fails again, then the process repeats or robots go to the origin to meet.
A process that involves $k$ random bits (and hence $2k$ possible meeting points) will be referred to as $k$-step 3-Markovian. 
Note that we allow $k=\infty$. 
The formal description of the algorithm is as follows.

\begin{algorithm}
\caption{$k$-RB$_{\beta,\gamma}$}
\label{kRB}
\begin{algorithmic}[1]
\STATE	Repeat $k$ times
\STATE ~~~~~~ Do a random $\beta$-darting.
\STATE ~~~~~~ Do a $\gamma$-darting in the opposite direction
\STATE Go to origin (if peer is not already met).
\end{algorithmic}
\end{algorithm}
 
 Observe that the algorithm of Theorem~\ref{thm: 3-mark easy} can be alternatively described as $\infty$-RB$_{\pi/2-\alpha/2,\pi/2-\alpha/2}$, while 1RB$_\beta$ is equivalent to $1$-RB$_{\beta,0}$. Next we analyze $k$-RB$_{\beta,\gamma}$ for all values of $k, \beta, \gamma$. 
Our goal is to analyze the expected rendezvous time, denoted by $\ren{k}{}{\beta}{\gamma}$. 
We adopt the language either of \srdr\ or of \srda\ depending on what is more convenient, in which case $\ren{k}{}{\beta}{\gamma}$ will be either a function of $\rho$ or of $\alpha$. To make this more explicit in our notation, and in order to remove any ambiguity, we will be writing $\ren{k}{\rho}{\beta}{\gamma}$ and $\ren{k}{\alpha}{\beta}{\gamma}$ for the expected running time in \srdr\ and \srda, respectively. Note that $\ren{k}{\rho}{\beta}{\gamma}=\frac{\ren{k}{\alpha}{\beta}{\gamma}}{\sinn\alpha}$.

 \begin{lemma}
\label{lem: krb performance}
For every fixed $\alpha$, the performance of 
$k$-RB$_{\beta,\gamma}$
for \srda, when $k=1,\infty$, is
\begin{align}
\ren{1}{\alpha}{\beta}{\gamma} &=\tfrac{1}{2} \csc (\alpha +\beta ) (\sin (\beta ) \csc (2 \alpha +\gamma ) (3 \sin (\alpha ) \cos (\alpha )+\sin (\gamma ))+2 \sin (\alpha ))
\label{equa: ren 1 value}
\\
\ren{\infty}{\alpha}{\beta}{\gamma} &=\frac{\sin (\alpha ) (3 \sin (\alpha -\beta )-3 \sin (\alpha +\beta )-4 \sin (2 \alpha +\gamma ))}{-2 \cos (\alpha -\beta +\gamma )+2 \cos (3 \alpha +\beta +\gamma )+\cos (\beta -\gamma )-\cos (\beta +\gamma )}.
\label{equa: ren inf value}
\end{align}
\end{lemma}

\begin{proof}
Note that each random step of $k$-RB$_{\beta,\gamma}$ involves two darting moves. 
For fixed $\alpha$, and a disk of radius 1, let $w=w(\beta)$ be the length of the line segment between the position of a robot and the possible meeting point at the bisector of the critical arcs in the first darting move. 
 Let also $y=y(\beta)$ denote the distance of the possible meeting point from the origin (see also Figure~\ref{figure 1F2B}). The values for $w,y$ are obtained as in the proof of Lemma~\ref{lem: 1rb performance} and are summarized below. 
 \begin{figure}
    \centering
       \includegraphics[width=8cm]{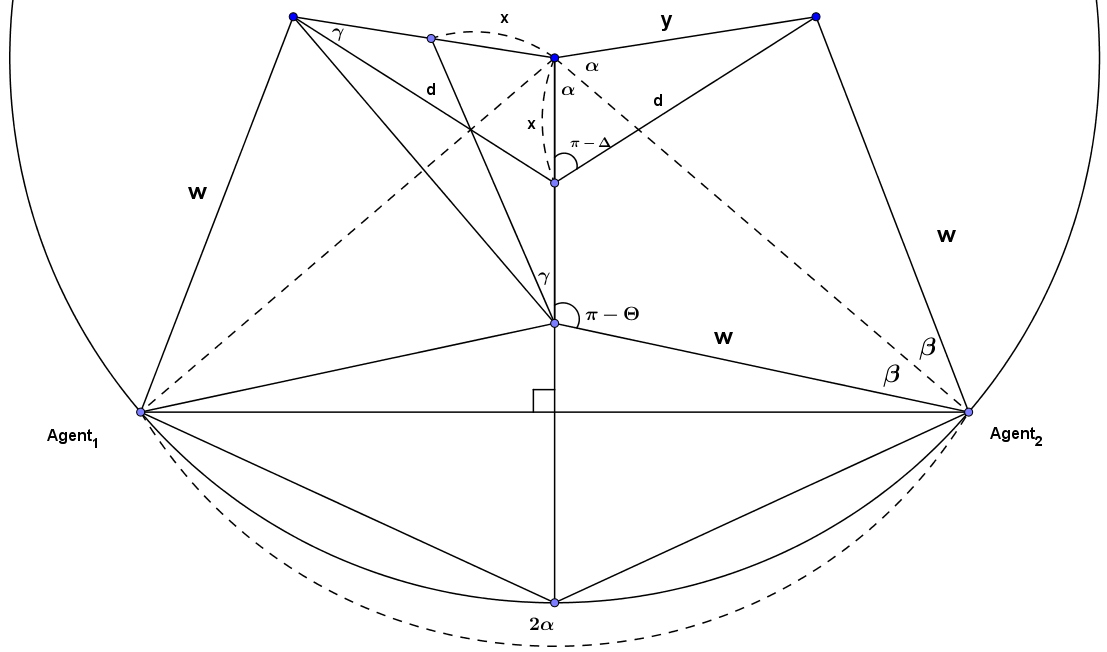}
 \caption{Geometry of Algorithm $k$-RB$_{\beta,\gamma}$, where $\Theta=\alpha+\beta$ and $\Delta=2\alpha+\gamma$. }
    \label{figure 1F2B}
\end{figure}
Notice that after the first darting move, robots are in a disk of radius $y$. 
Similarly, let $d=d(\beta, \gamma)$ be the length of the line segment between the position of a robot after the first darting attempt and the possible meeting point at the bisector of the critical arcs in the second darting move. Let also $x=x(\beta,\gamma)$ denote the distance of the second possible meeting point from the origin. Overall, we have 
\begin{align}
w &= \sinn\alpha\csc(\alpha+\beta)
\label{equa: value of w}
\\
y &= \sinn\beta\csc(\alpha+\beta)
\label{equa: value of y}
\\
x &= y \sinn\gamma\cscc{2\alpha+\gamma}
\label{equa: value of x}
\\
d &= y \sinn{2\alpha}\cscc{2\alpha+\gamma}.
\label{equa: value of d}
\end{align}
Now consider an iteration of the algorithm, where the radius of the disk is 1. The probability of the agents meeting at this iteration is $\tfrac{1}{2},$ and given that they meet the distance travelled is equally likely to be $w$ or $w+d$, giving a contribution to the mean cost equal to $w+\tfrac{1}{2}d$. If robots do not meet at this step, then they are at distance $x$ from the origin. So they either go to the origin, if number of iterations has exceeded $k$, or they repeat. Hence, 
\begin{align}
\ren{k}{\alpha}{\beta}{\gamma} &= \sum_{i=0}^{k-1}\left(\frac{1}{2}\right)^{i+1} \left[x^{i}(w +\tfrac{1}{2}d)+\sum_{j=0}^{i-1}(w+d)x^{j} \right]+\left(\frac{1}{2}\right)^k(x^k+\sum_{j=0}^{k-1}(w+d)x^{j})\nonumber \\
&= \sum_{i=0}^{k-1}\left(\frac{1}{2}\right)^{i+1} \left[x^{i}(w+\tfrac{1}{2}d)+(w+d)\frac{x^{i+1}-x}{x^2-x}\right] \nonumber
+\left(\frac{1}{2}\right)^k \left(x^k+(w+d)\frac{x^k-1}{(x-1)}\right) \nonumber\\
&=\left(\frac{1}{2}\right)^{k+1} \frac{x^k (3 d+4 w+2 x-4)-2^k (3 d+4 w)}{x-2}, \nonumber
\end{align}
Setting $k=1$ and taking the limit $k\rightarrow \infty$ (note that $0<x<1$) gives
$$
\ren{1}{\alpha}{\beta}{\gamma} = w+\tfrac{3}{4}d+\tfrac{1}{2}x,
$$
and
$$
\ren{\infty}{\alpha}{\beta}{\gamma}= \frac{3d+4w}{2(2-x)}
.$$
Then, the statement of the Lemma follows after elementary trigonometric manipulations. 
\ignore{
w = Sin[\[Alpha]]*Csc[\[Alpha] + \[Beta]];
y = Sin[\[Beta]]*Csc[\[Alpha] + \[Beta]];
d = y*Sin[2*\[Alpha]]*Csc[2*\[Alpha] + \[Gamma]];
x = y*Sin[\[Gamma]]*Csc[2*\[Alpha] + \[Gamma]];
FullSimplify[TrigExpand[w + 3/4*d + 1/2*x]
 ]
Simplify[TrigExpand[(3*d + 4*w)/2/(2 - x)]
 ]
}
\end{proof}

\begin{theorem}
\label{thm: opt 1rbbetagamma}
Consider problem \srdr.
If $\rho < \csc \left(\frac{1}{2} \cos ^{-1}\left(\frac{2}{3}\right)\right) \approx 2.44949$, then the optimal $1$-RB$_{\beta,\gamma}$ algorithm is obtained for $\overline{\gamma}=0$, and the algorithm is identical to the optimal 1RB$_\beta$ algorithm (see Theorem~\ref{thm: opt 1rbbeta}). 

If $\rho \geq \csc \left(\frac{1}{2} \cos ^{-1}\left(\frac{2}{3}\right)\right)$, then the optimal $1$-RB$_{\beta,\gamma}$ is obtained for the following parameters
\begin{align*}
\overline{\gamma} &= \cos ^{-1}\left(\frac{2}{3}\right)-2 \sin ^{-1}\left(\frac{1}{\rho }\right) \\
\overline{\beta} &= \cos ^{-1}\left(\frac{3}{4} \cos \left(\cos ^{-1}\left(\frac{2}{3}\right)-2 \sin ^{-1}\left(\frac{1}{\rho }\right)\right)\right)-\sin ^{-1}\left(\frac{1}{\rho }\right)
\end{align*}
\ignore{
cbar[a_] := ArcCos[2/3] - 2 a;
bbar[a_] := ArcCos[3/4*Cos[cbar[a]]] - a;
cbar[ArcSin[1/\[Rho]]]
bbar[ArcSin[1/\[Rho]]]
}
For the optimal parameters, the algorithm has competitive ratio $\coss{\overline{\beta}}$ which equals
$$
\frac{1}{2} \left(-\frac{\sqrt{5}}{\rho ^2}+\sqrt{\rho ^2-1}-\frac{2 \sqrt{\rho ^2-1}}{\rho ^2}+2 \sqrt{1-\frac{\left(\rho  \left(\sqrt{5-\frac{5}{\rho ^2}}+\rho \right)-2\right)^2}{4 \rho ^4}}+\sqrt{5}\right)
$$
\ignore{
test[\[Rho]_] := 
 3/4 Sqrt[1 - 
    1/\[Rho]^2] (2/3 \[Rho] - 4/(3 \[Rho]) + (
     2 Sqrt[5] Sqrt[1 - 1/\[Rho]^2])/3) + Sqrt[
  1 - 9/16 (2/3 - 4/(3 \[Rho]^2) + (2 Sqrt[5] Sqrt[1 - 1/\[Rho]^2])/(
      3 \[Rho]))^2]
      }
and it is $5.32366$-effective. 
\ignore{
NSolve[ test[\[Rho]] == 425/100, \[Rho]]
}
\end{theorem}

\begin{proof}
Lemma~\ref{lem: krb performance} gives us 
performance $\ren{1}{\alpha}{\beta}{\gamma}$ 
of $1$-RB$_{\beta,\gamma}$
for problem \srda, and the competitive ratio is obtained by scaling by $\sinn{\alpha}$. 
The critical points of $\ren{1}{\alpha}{\beta}{\gamma}$ are $\overline{\beta}, \overline{\gamma}$ satisfying equations
\begin{align}
\overline{\gamma} &= \arccos\tfrac{2}{3} - 2\alpha \label{equa:opt gamma}\\ 
\overline{\beta}&= \arccos(\tfrac{3}{4}\cos(\overline{\gamma}))-\alpha \label{equa:opt beta},
\end{align}
as shown in Lemmata~\ref{Critical c for k1},~\ref{Alpern first step proof}. These equations have unique solutions in $[0,\frac{\pi}{2}-\alpha]$ if and only if $\alpha \leq \frac{1}{2}\arccos\frac{2}{3};$ otherwise, they have no solution.
As $\ren{1}{\alpha}{\beta}{\gamma}$ has at most one critical point and is locally convex at that point (see Lemma~\ref{proof of optimality first Alpern}), these equations minimize $\ren{1}{\alpha}{\beta}{\gamma}$ for  $\alpha < \frac{1}{2}\arccos\frac{2}{3}.$ 
Now we substitute~\eqref{equa:opt gamma} and~\eqref{equa:opt beta} in~\eqref{equa: ren 1 value} to obtain, after straightforward manipulations, that at it's minimum $\ren{1}{\alpha}{\overline{\beta}}{\overline{\gamma}}=\cos\overline{\beta}.$

For $\alpha \geq \frac{1}{2}\arccos\frac{2}{3}$, $\ren{1}{\alpha}{\beta}{\gamma}$ is monotone increasing with respect to $\gamma$, and thus is optimized at $\gamma = 0,$ at which the strategy becomes identical to the one-step algorithm described before. Optimal parameters and run times can be calculated accordingly for \srdr\ using transformation $\alpha=\arcsinn{1/\rho}$ and simplifying trigonometric expressions. 
\end{proof}

We can now compute also the optimal parameters for $\infty$-RB$_{\beta,\gamma}$. Since the competitive ratio becomes a lengthy expression in $\rho$ for \srdr, we choose to only comment on the effectiveness of the resulting algorithm. The competitive ratio will be explicit from our calculations. 

\begin{theorem}
\label{thm: opt parameters for inftyRB}
For all $\rho\geq 1/\sinn{1/2}\approx2.08583$, the optimal $\infty$-RB$_{\beta,\gamma}$ algorithm for \srdr\ uses parameters $\overline{\beta}, \overline{\gamma}$ satisfying equations
\begin{align}
\tfrac34\coss{\overline{\gamma}}&=\coss{\arcsinn{1/\rho}+\overline{\beta}}  \label{equa: cond 1}\\
\tfrac23\coss{\overline{\beta}}&=\coss{2\arcsinn{1/\rho}+\overline{\gamma}}. \label{equa: cond 2}
\end{align}
In particular, we have 
\begin{align}
\overline\beta :=& \arctan 
\left(
\frac{-v + \sqrt{v^2-(\tfrac{9}{4}\cos^2\alpha-1)(\tfrac{5}{4}-v^2)}}{\tfrac{9}{4}\cos^2\alpha-1}
\right) 
\label{equa: opt beta}
\\
\overline\gamma:= 
&\arccoss{\tfrac43\coss{\alpha+\overline\beta}}.
\label{equa: opt gamma}
\end{align}
where $v:= (2\cos\alpha-\cos2\alpha)\csc 2\alpha$ and $\alpha=\arcsinn{1/\rho}$.
The competitive ratio of the algorithm can be computed by substituting $\overline\beta, \overline\gamma$ in~\eqref{equa: ren inf value}. 
Also for these values of $\overline\beta, \overline\gamma$, 
the algorithm is 
$7.13678$-effective. 
\end{theorem}
\begin{proof}
For convenience we adopt the language of \srda. 
The nonlinear system~\eqref{equa: cond 1},~\eqref{equa: cond 2} characterizes the critical points of function $\ren{\infty}{\alpha}{\beta}{\gamma}:\reals^2\mapsto \reals$, i.e. it is obtained by requiring that 
$$\frac{\partial}{\partial \beta}\ren{\infty}{\alpha}{\beta}{\gamma}=\frac{\partial}{\partial \gamma}\ren{\infty}{\alpha}{\beta}{\gamma}=0.$$
 We prove this in Lemma~\ref{lem: conditions beta gamma}. 

Now observe that equations~\eqref{equa: cond 1}, \eqref{equa: cond 2} is just a system of polynomial equations in $\coss\beta, \coss\gamma$. In fact, substituting one for the other results in a degree 4 polynomial equation that can be solved analytically. Only one of the solutions satisfies conditions $0\leq \beta\leq \pi/2-\alpha$, which is the $\overline\beta=\overline\beta(\alpha)$ described in~\eqref{equa: opt beta}.
The value of $\overline\gamma$ is calculated using~\eqref{equa: cond 1} as 
$
\overline\gamma:=\arccoss{\tfrac43\coss{\alpha+\overline\beta}}.
$

For all $\alpha<3/4$, we show in Lemma~\ref{lem: good bounds opt beta gamma} that $0\leq \overline\beta, \overline\gamma \leq \pi/2-\alpha$. 
Finally, in Lemma~\ref{lem: hessian of inftystep} we show that, for all $\alpha<1/2$, the aforementioned values of $\overline\beta, \overline\gamma$ do indeed correspond to a minimizer for $\ren{\infty}{\alpha}{\beta}{\gamma}$ by showing that $\nabla^2\ren{\infty}{\alpha}{\overline\beta}{\overline\gamma}$ is positive definite. 

Overall, we conclude that $\overline\beta,\overline\gamma$ do minimize $\ren{\infty}{\alpha}{\beta}{\gamma}$, in which case the competitive ratio becomes $\ren{\infty}{\alpha}{\overline\beta}{\overline\gamma}/\sinn{\alpha}$. Equating the last expression with 4.25, and solving for $\rho=1/\sinn{\alpha}$ gives numerical value $\rho =7.13678$.
\end{proof}
  \ignore{
 w[\[Alpha]_, \[Beta]_] := Sin[\[Alpha]]*Csc[\[Alpha] + \[Beta]];
y[\[Alpha]_, \[Beta]_] := Sin[\[Beta]]*Csc[\[Alpha] + \[Beta]];
d[\[Alpha]_, \[Beta]_, \[Gamma]_] := 
  y[\[Alpha], \[Beta]]*Sin[2*\[Alpha]]*Csc[2*\[Alpha] + \[Gamma]];
x[\[Alpha]_, \[Beta]_, \[Gamma]_] := 
  y[\[Alpha], \[Beta]]*Sin[\[Gamma]]*Csc[2*\[Alpha] + \[Gamma]]; 
perf[\[Alpha]_, \[Beta]_, \[Gamma]_] := (3*
      d[\[Alpha], \[Beta], \[Gamma]] + 4*w[\[Alpha], \[Beta]])/
   2/(2 - x[\[Alpha], \[Beta], \[Gamma]])
va[a_] := (2*Cos[a] - Cos[2*a])*Csc[2*a]
barb[a_] := 
  ArcTan[ (-va[a] + 
      Sqrt[va[a]^2 - (9/4*Cos[a]^2 - 1)*(5/4 - va[a]^2)])/(9/4*
       Cos[a]^2 - 1)];
barg[a_] := ArcCos[2/3*Cos[barb[a]]] - 2*a;
exprendtime[\[Alpha]_] := 
  perf[\[Alpha], barb[\[Alpha]], barg[\[Alpha]]];
(*compratio=*)
\[Rho]*exprendtime[ArcSin[1/\[Rho]]]
Plot[ \[Rho]*exprendtime[ArcSin[1/\[Rho]]], {\[Rho], 1.5, 20}]
FindRoot[ \[Rho]*exprendtime[ArcSin[1/\[Rho]]] == 425/100, {\[Rho], 
  10}]
}

We conclude this section by providing some asymptotic analysis for the optimal parameters $\overline\beta, \overline\gamma$ of Algorithm $\infty$-RB$_{\beta,\gamma}$ as $\rho\rightarrow\infty$. As expected, both $\overline\beta, \overline\gamma$ tend to $\pi/2$, as well as $\ren{\infty}{\rho}{\overline\beta}{\overline\gamma}$ tends to 5 (the competitive ratio of the \srl\ algorithm we are extending). This is what we make explicit with the next theorem, by also providing the rate of convergence. 

\begin{theorem}
\label{thm: asymptotics rendezvous}
For the optimal parameters $\overline\beta=\overline\beta(\rho), \overline\gamma=\overline\gamma(\rho)$ of Algorithm $\infty$-RB$_{\beta,\gamma}$, we have 
\begin{align}
\lim_{\rho\rightarrow \infty} \frac{\pi/2-\overline\beta}{\arcsinn{1/\rho}}&=5
\label{equa: beta limit behavior}
\\
\lim_{\rho\rightarrow \infty} \frac{\pi/2-\overline\gamma}{\arcsinn{1/\rho}}&=\frac{16}3
\label{equa: gamma limit behavior}
\end{align}
Moreover, 
$$
\lim_{\rho\rightarrow \infty} \rho^2 (5-\ren{\infty}{\rho}{\overline\beta}{\overline\gamma}) = 289/6. 
$$
\end{theorem}
 
 \begin{proof}
We use the language of \srda, and in particular we consider $\overline\beta=\overline\beta(\alpha), \overline\gamma=\overline\gamma(\alpha)$, and $\alpha\rightarrow 0$. 
 By Theorem~\ref{thm: opt parameters for inftyRB}, and using~\eqref{equa: opt beta}, it is easy to see that 
$
\lim_{\alpha\rightarrow 0} \overline\beta(\alpha)
=
\pi/2
$. Some straightforward but tedious calculations also show that 
$
\lim_{\alpha\rightarrow 0} \frac{\pi/2-\overline\beta(\alpha)}{\alpha}=5$. 
The statement for $\overline\gamma$ follows similarly using again 
Theorem~\ref{thm: opt parameters for inftyRB}
and in particular that $\overline\gamma=\arccoss{\tfrac43\coss{\alpha+\overline\beta}}$. 

Note that the expected rendezvous time $\ren{\infty}{\rho}{\overline\beta}{\overline\gamma}$ in \srdr\ is also the competitive ratio of the problem. In the language of \srda\, the competitive ratio is $\ren{\infty}{\alpha}{\overline\beta}{\overline\gamma}/\sinn{\alpha}$.

By~\eqref{equa: beta limit behavior} we know that as $\alpha$ tends to 0, $\overline\beta$ behaves similar to $\pi/2-5\alpha$, and by~\eqref{equa: gamma limit behavior}
that $\overline\gamma$ behaves similar to $\pi/2-\frac{16}3\alpha$. 
Using now~\eqref{equa: ren inf value} of Lemma~\ref{lem: krb performance}, we have that 
\begin{align*}
\lim_{\alpha\rightarrow 0}
\frac{
\ren{\infty}{\alpha}{\overline\beta}{\overline\gamma}
}
{
\sinn{\alpha}
}
&=
\lim_{\alpha\rightarrow 0}
\frac{
\ren{\infty}{\alpha}{\pi/2-5\alpha}{\pi/2-\tfrac{16}3\alpha}
}
{
\sinn{\alpha}
}
\\
&=
\lim_{\alpha\rightarrow 0}
\frac{-4 \cos \left(\frac{10 \alpha }{3}\right)-3 (\cos (4 \alpha )+\cos (6 \alpha ))}{\cos \left(\frac{\alpha }{3}\right)-2 \left(\cos \left(\frac{2 \alpha }{3}\right)+\cos \left(\frac{22 \alpha }{3}\right)\right)+\cos \left(\frac{31 \alpha }{3}\right)}
=5.
\end{align*}
Similarly, we can can show that 
$
\lim_{\alpha\rightarrow 0} \frac{5-\frac{\ren{\infty}{\alpha}{\overline\beta}{\overline\gamma}}{\sinn{\alpha}}}
{\sin^2\left(\alpha\right)}
 = 289/6. 
$
\end{proof}

\section{Energy-Efficient Rendezvous}
\label{sec: energy}

\subsection{Energy Analysis of our Infinite-Step Rendezvous Algorithm}
\label{sec: energy 3-markovian infty}

A unique feature of the \srd\ problem is that, unlike in \srl, the worst case rendezvous time can be finite.
As before we distinguish whether we calculate the energy of $\infty$-RB$_{\beta,\gamma}$
in \srdr\ or in \srda\ by writing 
$\ene{\infty}{\rho}{\beta}{\gamma}$ and $\ene{\infty}{\alpha}{\beta}{\gamma}$, respectively.

\begin{lemma}
\label{lem: energy bounded}
The energy $\ene{\infty}{\alpha}{\beta}{\gamma}$ of $\infty$-RB$_{\beta,\gamma}$ for \srda\ is finite if and only if 
$\sinn\beta \sinn \gamma<\sinn{\alpha+\beta}\sinn{2\alpha+\gamma}.$
Moreover
\begin{equation}
\label{equa: energy for opt infty step}
\ene{\infty}{\alpha}{\beta}{\gamma}:=
\frac{\sinn\alpha\csc(\alpha+\beta)+
\sinn\beta\csc(\alpha+\beta)
\sinn{2\alpha}\cscc{2\alpha+\gamma}}
{1-
\sinn\beta\csc(\alpha+\beta)
\sinn\gamma\cscc{2\alpha+\gamma}}
.
\end{equation}
\end{lemma}

\begin{proof}[ of Lemma~\ref{lem: energy bounded}]
For convenience, we analyze the performance for \srda. As in the proof of Lemma~\ref{lem: krb performance} (see also Figure~\ref{figure 1F2B}), in every iteration of $\infty$-RB$_{\beta,\gamma}$ agents walk a distance equal to 
$w+d$
and the radius of their disk is shrunk by $x$, so that the energy of $\infty$-RB$_{\beta,\gamma}$ is calculated as 
$$(w+d) \sum_{j=0}^\infty x^{j},$$
where 
$w,x,d$ are as 
in~\eqref{equa: value of w},
\eqref{equa: value of x}
and \eqref{equa: value of d}, respectively. 
Clearly, the sum of the energy converges if and only if $x<1$, or equivalently 
$$
\frac{\sinn\beta \sinn \gamma}{\sinn{\alpha+\beta}\sinn{2\alpha+\gamma}}<1.
$$
When the energy sum converges, it equals 
$$
\frac{w+d}{1-x}.
$$
Now we use~\eqref{equa: value of w}, \eqref{equa: value of x}, \eqref{equa: value of d}, and the claim follows. 
\end{proof}

\begin{lemma}
\label{lem: energy bounded for opt alg}
For any fixed $\rho$, the energy $\ene{\infty}{\rho}{\overline\beta}{\overline\gamma}$ of the optimal $\infty$-RB$_{\overline\beta,\overline\gamma}$ 
is finite. 
\end{lemma}

\begin{proof}
Translating Theorem~\ref{thm: opt parameters for inftyRB} to the language of \srda\, we know that parameters $\overline\beta, \overline\gamma$ satisfy 
$$
\tfrac34\coss{\overline{\gamma}}
=\coss{\alpha+\overline{\beta}}$$
and 
$$
\tfrac23\coss{\overline{\beta}}
=\coss{2\alpha+\overline{\gamma}} 
$$
or equivalently that 
$$
\sinn{\alpha+\overline{\beta}} 
=
\sqrt{
1-\tfrac9{16}\cos^2\left(\overline{\gamma}\right)
},$$
and
$$
\sinn{2\alpha+\overline{\gamma}} 
=
\sqrt{
1-\tfrac49\cos^2\left(\overline{\beta}\right)
} .
$$
But then, it is immediate that 
$\sinn{\alpha+\overline\beta} > \sinn{\overline\gamma}$ and that 
$\sinn{2\alpha+\overline\gamma} > \sinn{\overline\beta}$. Multiplying side-wise the latter two inequalities shows that the condition of Lemma~\ref{lem: energy bounded} is satisfied. Hence, the energy of $\infty$-RB$_{\overline\beta,\overline\gamma}$ is finite for every $\rho$. 
\end{proof}

Using values $\overline\beta, \overline\gamma$ (see~\eqref{equa: opt beta} and~\eqref{equa: opt gamma} of Theorem~\ref{thm: opt parameters for inftyRB}), and substituting in~\eqref{equa: energy for opt infty step} of Lemma~\ref{lem: energy bounded} we obtain an explicit, yet complicated, function of $\alpha$ (or equivalently of $\rho=1/\sinn\alpha$) for $\ene{\infty}{\alpha}{\overline\beta}{\overline\gamma}$. Using \textsc{Mathematica} we can observe graphically that $\ene{\infty}{\rho}{\overline\beta}{\overline\gamma}$ is strictly increasing (which is also expected), and that $\ene{\infty}{\rho}{\overline\beta}{\overline\gamma}/\rho^2$ is strictly decreasing in $\rho>2$. 
However a formal proof is eluding us due to the complication of the formulas. Nevertheless, we can find the asymptotic behaviour of the energy as $\rho$ tends to infinity.

\begin{theorem}
\label{thm: asymptotics energy}
For the optimal parameters $\overline\beta=\overline\beta(\rho), \overline\gamma=\overline\gamma(\rho)$ of Algorithm $\infty$-RB$_{\beta,\gamma}$, we have 
$$
\lim_{\rho\rightarrow \infty} \frac{
 \ene{\infty}{\rho}{\overline\beta}{\overline\gamma}}
 {\rho^2} = \frac{18}{79}. 
$$
\end{theorem}
 
\begin{proof}
We adopt the language of \srda, and we invoke Theorem~\ref{thm: asymptotics rendezvous}. 
By~\eqref{equa: beta limit behavior} we know that as $\alpha$ tends to 0, $\overline\beta$ behaves similar to $\pi/2-5\alpha$, and by~\eqref{equa: gamma limit behavior}
that $\overline\gamma$ behaves similar to $\pi/2-\frac{16}3\alpha$. 

Using now~\eqref{equa: ren inf value} of Lemma~\ref{lem: krb performance}, we have that 
\begin{align*}
\lim_{\alpha\rightarrow 0}
\ene{\infty}{\alpha}{\overline\beta}{\overline\gamma}
\sinn{\alpha}
&=
\lim_{\alpha\rightarrow 0}
\ene{\infty}{\alpha}{\pi/2-5\alpha}{\pi/2-\tfrac{16}3\alpha}
\sinn{\alpha}
\\
&
=
\lim_{\alpha\rightarrow 0}
\frac{\sin (\alpha ) \left(\sin (\alpha )+\sin (2 \alpha ) \cos (5 \alpha ) \sec \left(\frac{10 \alpha }{3}\right)\right)}{\cos (4 \alpha )-\cos (5 \alpha ) \cos \left(\frac{16 \alpha }{3}\right) \sec \left(\frac{10 \alpha }{3}\right)}.
\end{align*}
The latter limit can be computed in \textsc{Mathematica} and it equals 18/79. 
\end{proof}

An immediate corollary of Theorem~\ref{thm: asymptotics energy} is that $\ene{\infty}{\rho}{\overline\beta}{\overline\gamma}=\Theta(\rho^2)$. As long as the rendezvous between the two agents is not realized, both follow random-walk-like trajectories (see Figure~\ref{fig: spiral}).
\begin{figure}[h!]
\centering
\includegraphics[width=.8\linewidth]{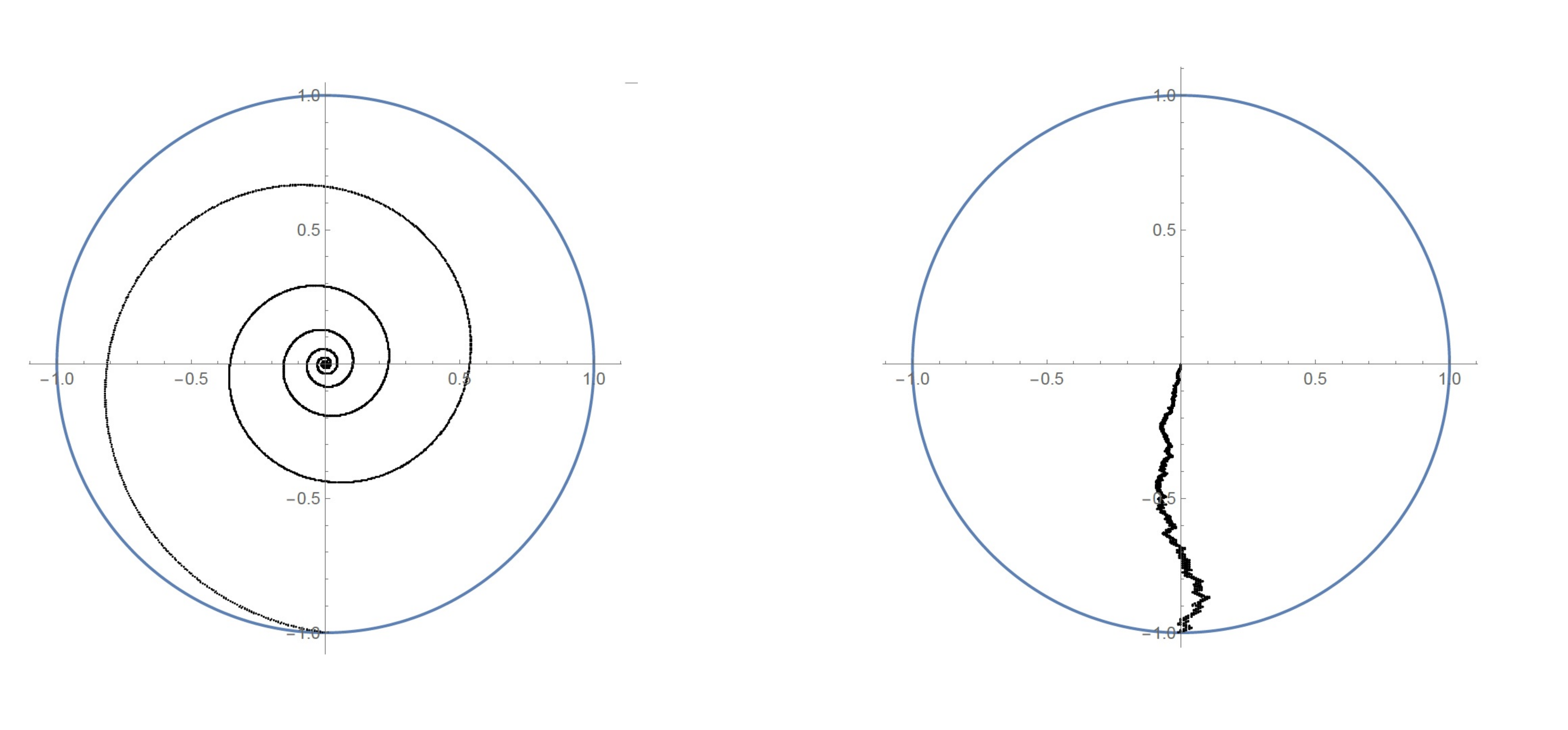}
\caption{Possible trajectories of one agent in Algorithm $\infty$-RB$_{\beta,\gamma}$ in \srda\ when $\alpha=0.01$. The figure on the left depicts the trajectory in which the agent always attempts to meet her peer first by moving ccw and then cw, resulting in a spiral. The figure on the right depicts a random trajectory. Both trajectories have the same length which is approximately $22.7911$.}
\label{fig: spiral}
\end{figure}

\subsection{Expected Rendezvous Time - Energy Tradeoffs}
\label{sec: constrained}

In this section we attempt to understand how energy constraints can impact the performance of $\infty$-RB$_{\beta,\gamma}$.
By Theorem~\ref{thm: asymptotics rendezvous} we know that the 
optimal $\infty$-RB$_{\overline\beta,\overline\gamma}$ Algorithm induces competitive ratio 5, asymptotically in $\rho\rightarrow \infty$. By Theorem~\ref{thm: asymptotics energy} we know that the same algorithm (with the same parameters) requires $\Theta\left(\rho^2 \right)$ energy. 
In the other extreme, if the energy is less that $\rho$, then the problem admits no solution (and if the energy equals $\rho$, then the best rendezvous is attained when robots go directly to the reference point). Hence, we are motivated to study the problem of minimizing the expected rendezvous time in \srdr\ given that agents' energy is between $\rho$ and $\frac{18}{79}\rho^2$. 
Somehow surprisingly, we show below that for every $\epsilon>0$ we can preserve a competitive ratio of 5 and energy no more than $\epsilon \rho^2+o(\rho^2)$
or competitive ratio $5+\epsilon$ and energy no more than $\tfrac2{\sqrt{\epsilon}} \rho+o(\rho)$, both asymptotically in $\rho$. 

\begin{theorem}
\label{thm: main energy}
The following claims are true asymptotically for \srdr\ as $\rho \rightarrow \infty$. 
For every $\epsilon>0$, there exist $\beta_1, \gamma_1$ so that the competitive ratio of $\infty$-RB$_{\beta_1,\gamma_1}$ is 5, as well as $\ene{\infty}{\rho}{\beta_1}{\gamma_1}/\rho^2 \leq \epsilon$. 
Moreover, 
for every $\delta>0$, there exist $\beta_2, \gamma_2$ so that the competitive ratio of $\infty$-RB$_{\beta_2,\gamma_2}$ is $5+\delta$, as well as $\ene{\infty}{\rho}{\beta_1}{\gamma_1}/\rho \leq 2/\sqrt{\delta}$. 
\end{theorem}

The two claims of Theorem~\ref{thm: main energy} follow directly from the two lemmata below. 
In particular, Lemma~\ref{lem: compratio 5, energy r^2} shows that the competitive ratio of $\infty$-RB$_{\beta,\gamma}$ can stay 5, even if the energy needed to solve \srdr\ is $\epsilon\rho^2$, for each $\epsilon>0$. 
Lemma~\ref{lem: compratio 5+e, energy r*1/sqrt[e]} shows that if one is willing to have competitive ratio $5+\epsilon$, then that would be possible with linear energy in $\rho$, and in particular no more than $2\rho/\sqrt{\epsilon}$, again for every $\epsilon>0$.

\begin{lemma}
\label{lem: compratio 5, energy r^2}
For every positive $\epsilon>0$, there exist $\beta, \gamma$, such that for the performance of $\infty$-RB$_{\beta,\gamma}$ for \srdr, we have that 
$$
\lim_{\rho\rightarrow \infty}
\rho \left(
\ren{\infty}{\rho}{\beta}{\gamma} 
-5\right)
= \frac{27}{11\epsilon^2}-\Theta(1/\epsilon) 
,$$ and
$$
\lim_{\rho\rightarrow \infty}
\frac{
\ene{\infty}{\rho}{\beta}{\gamma}
}
{\rho^2} 
= \epsilon.
$$
\end{lemma}

\begin{proof}
We use the language of \srda. 
For some positive constants $k,m$, we use $\beta=\pi/2-k\alpha$ and $\gamma=\pi/2-m\alpha$. 
First, using Lemma~\ref{lem: krb performance}, it is easy to see that 
\begin{align*}
&\lim_{\alpha\rightarrow 0} 
\frac{\ren{\infty}{\alpha}{\beta}{\gamma}}{\sinn{\alpha}}
 \\
 &=
\lim_{\alpha\rightarrow 0}
\frac{-6 \cos (\alpha ) \cos (\alpha  k)-4 \cos (\alpha  (m-2))}{\cos (\alpha  (k-m))-2 (\cos (\alpha  (k-m+1))+\cos (\alpha  (k+m-3)))+\cos (\alpha  (k+m))} \\
&=5. 
\end{align*}
Some more elaborate calculations can show in fact that 
\begin{equation}
\label{equa: conv compratio limit}
\lim_{\alpha\rightarrow 0}
\frac{
\frac{\ren{\infty}{\alpha}{\beta}{\gamma}}{\sinn{\alpha}}-5
}
{\sin^2\left(\alpha\right)}
=
k^2-10 k+\frac{3 m^2}{2}-16 m+\frac{39}{2}
\end{equation}

By the proof of Lemma~\ref{lem: energy bounded}
\begin{align}
&\lim_{\alpha\rightarrow 0}
\sinn\alpha
\ene{\infty}{\alpha}{\beta}{\gamma}  \notag  \\
&=
\lim_{\alpha\rightarrow 0}
\frac{\sin ^2(\alpha ) \csc \left(\frac{\alpha }{2}\right) (\cos (\alpha  (k+1))+\cos (\alpha -\alpha  k)+\cos (\alpha  (m-2)))}
{
2 \cos \left(\frac{\alpha }{2}\right) \sin (\alpha  (k+m-1))
-
2 \cos (\alpha -\alpha  k) \sin \left(\alpha  \left(\frac{3}{2}-m\right)\right)
}
\notag \\
&=
\frac{6}{2 k+4 m-5}.
\label{equa: conv energy limit} 
\end{align}

The claim follows by choosing $k=\frac{31 \epsilon +18}{22 \epsilon }$ and $m=\frac{6 \epsilon +12}{11 \epsilon }$. These values are obtained by requiring that 
\eqref{equa: conv energy limit} equals $\epsilon$ and minimizing~\eqref{equa: conv compratio limit}.
In particular, substituting $m,k$ in~\eqref{equa: conv compratio limit} we obtain 
$$
\lim_{\alpha\rightarrow 0}
\frac{
\frac{\ren{\infty}{\alpha}{\beta}{\gamma}}{\sinn{\alpha}}-5
}
{\sin^2\left(\alpha\right)}
=\frac{27}{11 \epsilon ^2}-\frac{237}{11 \epsilon }-\frac{39}{44}.
$$
Finally, substituting $m,k$ in~\eqref{equa: conv energy limit} we obtain 
$$
\lim_{\alpha\rightarrow 0}
\sinn\alpha
\ene{\infty}{\alpha}{\beta}{\gamma}
=\epsilon.
$$
\end{proof}

\begin{lemma}
\label{lem: compratio 5+e, energy r*1/sqrt[e]}
For every positive $\epsilon>0$, there exist $\beta, \gamma$, such that for the performance of $\infty$-RB$_{\beta,\gamma}$ for \srdr, we have that 
$$
\lim_{\rho\rightarrow \infty}
\ren{\infty}{\rho}{\beta}{\gamma} 
=5+\epsilon 
,$$ 
and
$$
\lim_{\rho\rightarrow \infty}
\frac{
\ene{\infty}{\rho}{\beta}{\gamma}
}
{\rho} 
\leq 
\frac2{\sqrt{\epsilon}}
.$$
\end{lemma}

\begin{proof}
We start with some simple observations. 
Using Lemma~\ref{lem: krb performance} we have
$$
\lim_{\alpha\rightarrow 0} 
\frac{
\ren{\infty}{\alpha}{\beta}{\gamma}
}
{\sinn\alpha}
=
\lim_{\alpha\rightarrow 0} 
\frac{\csc (\alpha +\beta ) (3 \cos (\alpha ) \sin (\beta ) \csc (2 \alpha +\gamma )+2)}
{2-\sin (\beta ) \sin (\gamma ) \csc (\alpha +\beta ) \csc (2 \alpha +\gamma )}
=
\frac2{\sinn\beta}+\frac3{\sinn{\gamma}}
$$

By the proof of Lemma~\ref{lem: energy bounded} and after simple manipulations, we have
\begin{align*}
\lim_{\alpha\rightarrow 0} 
\ene{\infty}{\alpha}{\beta}{\gamma}
=&
\lim_{\alpha\rightarrow 0} 
\frac{2 \cos \left(\frac{\alpha }{2}\right) (-\sin (\alpha -\beta )+\sin (\alpha +\beta )+\sin (2 \alpha +\gamma ))}
{(2 \cos (\alpha )+1) \sin \left(\frac{3 \alpha }{2}+\beta +\gamma \right)-\sin \left(\frac{\alpha }{2}-\beta +\gamma \right)}
\\
=&
\frac{
\sinn\gamma+2\sinn\beta
}{
\coss\beta\sinn\gamma+2\coss\gamma\sinn\beta
}
\end{align*}

Now use abbreviation $b=\sinn\beta$ and $c=\sinn\gamma$. Set $b=c=\frac{5}{5+\epsilon}$, and observe that 
$$
\lim_{\alpha\rightarrow 0} 
\frac{
\ren{\infty}{\alpha}{\beta}{\gamma}
}
{\sinn\alpha}
=5+\epsilon
$$
while 
$$
\lim_{\alpha\rightarrow 0} 
\ene{\infty}{\alpha}{\beta}{\gamma}
=
\frac{\epsilon +5}{\sqrt{\epsilon  (\epsilon +10)}}
\leq \frac2{\sqrt{\epsilon}}.
$$
Finally, we note that we can achieve the same competitive ratio, and slightly improve the required energy. For this, we need to set alternatively $b=\frac{2}{\lambda  \epsilon +2}$ and $c=\frac{3}{-\lambda  \epsilon +\epsilon +3}$. For each $\lambda\in [0,1]$ it is easy to see that $2/b+3/c=5+\epsilon$. 
Choosing also $\lambda=3/11$ minimizes the energy, which becomes 
$$
\lim_{\alpha\rightarrow 0} 
\ene{\infty}{\alpha}{\beta}{\gamma}
=
\frac{41 \epsilon +198}{3 \sqrt{3} \sqrt{\epsilon  (3 \epsilon +44)}+16 \sqrt{\epsilon  (4 \epsilon +33)}}.
$$
\end{proof}

\section{Conclusion}
We introduced and studied a new geometric variant of symmetric rendezvous that we call Symmetric Rendezvous in a Disk (\srd). Our main contribution pertains to the algorithmic reduction of known suboptimal algorithms for the classic Symmetric Rendezvous problem on a Line (\srl) to \srd. Since \srd\ can also be interpreted as a variant of \srl\ in which agents are equipped with additional advice, our results demonstrate how this advice can be beneficial to the expected rendezvous time, beating in some cases the conjectured best possible time for \srl. 
Special to \srd\ is also that, unlike in \srl, our algorithms induce bounded worst case (energy) performance. Motivated by this, we also studied energy-efficiency tradeoffs, and we showed that, somehow surprisingly, one can achieve rendezvous with limited energy (and with probability 1) by compromising only slightly on the expected rendezvous time. 

Our techniques can be generalized for all known improved rendezvous protocols for \srl, however optimal reductions will be challenging to obtain. Nevertheless, it is interesting to investigate heuristic reductions, which we leave as an open research direction. Other interesting variants of our problem include the introduction of more agents, or relaxations of the notion of advice that we are using.

\bibliographystyle{plain}
\bibliography{real_bibliography}

\begin{thebibliography}{10}

\bibitem{Alpern76}
Steve Alpern.
\newblock Hide and seek games.
\newblock Seminar, 1976.

\bibitem{Alpern1995}
Steve Alpern.
\newblock The rendezvous search problem.
\newblock {\em SIAM Journal on Control and Optimization}, 33(3):673--11, 05
  1995.

\bibitem{Alpern2002}
Steve Alpern.
\newblock Rendezvous search: A personal perspective.
\newblock {\em Operations Research}, 50(5):772--795, 2002.

\bibitem{Alpern2013}
Steve Alpern.
\newblock {\em Ten Open Problems in Rendezvous Search}, pages 223--230.
\newblock Springer New York, New York, NY, 2013.

\bibitem{alpern2005rendezvous}
Steve Alpern and Vic Baston.
\newblock Rendezvous on a planar lattice.
\newblock {\em Operations research}, 53(6):996--1006, 2005.

\bibitem{alpern2006common}
Steve Alpern and Vic Baston.
\newblock A common notion of clockwise can help in planar rendezvous.
\newblock {\em European journal of operational research}, 175(2):688--706,
  2006.

\bibitem{alpern2006rendezvous}
Steve Alpern and Vic Baston.
\newblock Rendezvous in higher dimensions.
\newblock {\em SIAM Journal on Control and Optimization}, 44(6):2233--2252,
  2006.

\bibitem{Alpern1995RSL}
Steve Alpern and Shmuel Gal.
\newblock Rendezvous search on the line with distinguishable players.
\newblock {\em SIAM Journal on Control and Optimization}, 33(4):1270--1276,
  July 1995.

\bibitem{searchGamesBook}
Steve Alpern and Shmuel Gal.
\newblock {\em The Theory of Search Games and Rendezvous.}
\newblock Number Vol. 55 in International Series in Operations Research \&
  Management Science. Springer, 2003.

\bibitem{alpern2002rendezvous}
Steve Alpern and Wei~Shi Lim.
\newblock Rendezvous of three agents on the line.
\newblock {\em Naval Research Logistics (NRL)}, 49(3):244--255, 2002.

\bibitem{ACCLPV12}
Julian Anaya, J{\'e}r{\'e}mie Chalopin, Jurek Czyzowicz, Arnaud Labourel,
  Andrzej Pelc, and Yann Vax{\`e}s.
\newblock Collecting information by power-aware mobile agents.
\newblock In {\em DISC}, volume 7611 of {\em LNCS}, pages 46--60. Springer,
  2012.

\bibitem{discreteLocations}
E.~J. Anderson and R.~R. Weber.
\newblock The rendezvous problem on discrete locations.
\newblock {\em Journal of Applied Probability}, 27(4):839--851, 1990.

\bibitem{lineWeirdPlayers}
Edward~J. Anderson and Skander Essegaier.
\newblock Rendezvous search on the line with indistinguishable players.
\newblock {\em SIAM Journal on Control and Optimization}, 33(6):1637--1642,
  1995.

\bibitem{BFFS07}
Lali Barri{\`e}re, Paola Flocchini, Pierre Fraigniaud, and Nicola Santoro.
\newblock Rendezvous and election of mobile agents: Impact of sense of
  direction.
\newblock {\em Theory Comput. Syst}, 40(2):143--162, 2007.

\bibitem{Baston99}
V.~J. Baston.
\newblock Two rendezvous search problems on the line.
\newblock {\em Naval Research Logistics}, 46:335--340, 1999.

\bibitem{baston1998rendezvous}
Vic Baston and Shmuel Gal.
\newblock Rendezvous on the line when the players' initial distance is given by
  an unknown probability distribution.
\newblock {\em SIAM Journal on Control and Optimization}, 36(6):1880--1889,
  1998.

\bibitem{beveridge2011symmetric}
Andrew Beveridge, Deniz Ozsoyeller, and Volkan Isler.
\newblock Symmetric rendezvous on the line with an unknown initial distance.
\newblock Technical Report, 2011.

\bibitem{CT04}
Elizabeth~J. Chester and Reha~H. T{\"u}t{\"u}nc{\"u}.
\newblock Rendezvous search on the labeled line.
\newblock {\em Operations Research}, 52(2):330--334, 2004.

\bibitem{collins2011synchronous}
Andrew Collins, Jurek Czyzowicz, Leszek Gasieniec, Adrian Kosowski, and
  Russell~A Martin.
\newblock Synchronous rendezvous for location-aware agents.
\newblock In {\em DISC}, volume 6950, pages 447--459. Springer, 2011.

\bibitem{CFR09}
Colin Cooper, Alan~M. Frieze, and Tomasz Radzik.
\newblock Multiple random walks and interacting particle systems.
\newblock In {\em ICALP}, volume 5556 of {\em LNCS}, pages 399--410. Springer,
  2009.

\bibitem{czyzowicz2008power}
Jurek Czyzowicz, Stefan Dobrev, Evangelos Kranakis, and Danny Krizanc.
\newblock The power of tokens: rendezvous and symmetry detection for two mobile
  agents in a ring.
\newblock {\em LNCS}, 4910:234--246, 2008.

\bibitem{CPL12}
Jurek Czyzowicz, Andrzej Pelc, and Arnaud Labourel.
\newblock How to meet asynchronously (almost) everywhere.
\newblock {\em ACM Trans. Algorithms}, 8(4):37:1--37:14, 2012.

\bibitem{Das2007}
Shantanu Das.
\newblock {\em Distributed computing with mobile agents: solving rendezvous and
  related problems}.
\newblock PhD thesis, University of Ottawa (Canada), 2007.

\bibitem{Das08c}
Shantanu Das.
\newblock Mobile agent rendezvous in a ring using faulty tokens.
\newblock In {\em Distributed Computing and Networking (9th ICDCN'08)}, volume
  4904 of {\em LNCS}, pages 292--297. Springer-Verlag (New York), Kolkata,
  India, January 2008.

\bibitem{DLM15}
Shantanu Das, Flaminia~L. Luccio, and Euripides Markou.
\newblock Mobile agents rendezvous in spite of a malicious agent.
\newblock In {\em ALGOSENSORS}, volume 9536 of {\em LNCS}, pages 211--224.
  Springer, 2015.

\bibitem{feinerman2014fast}
Ofer Feinerman, Amos Korman, Shay Kutten, and Yoav Rodeh.
\newblock Fast rendezvous on a cycle by agents with different speeds.
\newblock In {\em International Conference on Distributed Computing and
  Networking}, pages 1--13. Springer, 2014.

\bibitem{flocchini2004multiple}
Paola Flocchini, Evangelos Kranakis, Danny Krizanc, Nicola Santoro, and Cindy
  Sawchuk.
\newblock Multiple mobile agent rendezvous in a ring.
\newblock In {\em LATIN}, volume~4, pages 599--608. Springer, 2004.

\bibitem{hanetal}
Qiaoming Han, Donglei Du, Juan Vera, and Luis~F. Zuluaga.
\newblock Improved bounds for the symmetric rendezvous value on the line.
\newblock {\em Operations Research}, 56(3):772--782, 2008.

\bibitem{kranakis2006mobile}
Evangelos Kranakis, Danny Krizanc, and Euripides Markou.
\newblock Mobile agent rendezvous in a synchronous torus.
\newblock In {\em LATIN}, pages 653--664. Springer, 2006.

\bibitem{KKM10}
Evangelos Kranakis, Danny Krizanc, and Euripides Markou.
\newblock {\em The Mobile Agent Rendezvous Problem in the Ring}.
\newblock Synthesis Lectures on Distributed Computing Theory. Morgan \&
  Claypool Publishers, 2010.

\bibitem{kranakis2003mobile}
Evangelos Kranakis, Nicola Santoro, Cindy Sawchuk, and Danny Krizanc.
\newblock Mobile agent rendezvous in a ring.
\newblock In {\em Distributed Computing Systems}, pages 592--599. IEEE, 2003.

\bibitem{pelc2012deterministic}
Andrzej Pelc.
\newblock Deterministic rendezvous in networks: A comprehensive survey.
\newblock {\em Networks}, 59(3):331--347, 2012.

\bibitem{Prencipe07}
Giuseppe Prencipe.
\newblock Impossibility of gathering by a set of autonomous mobile robots.
\newblock {\em Theor. Comput. Sci}, 384(2-3):222--231, 2007.

\bibitem{TZ14}
Amnon Ta-Shma and Uri Zwick.
\newblock Deterministic rendezvous, treasure hunts, and strongly universal
  exploration sequences.
\newblock {\em ACM Trans. Algorithms}, 10(3):12:1--12:15, 2014.

\bibitem{uthaisombut2006symmetric}
Patchrawat~Patch Uthaisombut.
\newblock Symmetric rendezvous search on the line using move patterns with
  different lengths.
\newblock Working paper, 2006.

\end{thebibliography}

\appendix

\section{Omitted Lemmata (and Their Proofs)}
\label{appendix: proofs}

\begin{lemma}
\label{Critical c for k1}
Let $(\overline \beta,\overline \gamma)$ be a critical point of $\mathcal{R}_1^{\alpha}$, and set $\Delta=2\alpha+\overline\gamma$. Then $\cos\Delta =\tfrac{2}{3}.$
\end{lemma}

\begin{proof}
Let  $\overline{\beta},\overline{\gamma}$ be a critical point of $\mathcal{R}_1^{\alpha}$. Then $$\frac{\partial\mathcal{R}_1^{\alpha}(\overline{\beta},\overline{\gamma})}{\partial\gamma} = \frac{d\csc\Delta}{4}\left(-3\cos\Delta + 2\right) = 0,$$ and thus as $d$ is nonzero, it must be true that $\cos\Delta=\frac{2}{3}.$
\end{proof}

\begin{lemma}
\label{Alpern first step proof}
Any critical point $\overline{\beta},\overline{\gamma}$ of $\mathcal{R}_1^{\alpha}$ satisfies 
\begin{equation}
\cos(\Theta)=\frac{3}{4}\cos(\overline{\gamma}),
\label{bestbetaE1}
\end{equation}
where $\Theta=\alpha+\overline\beta$.
\end{lemma}
\begin{proof}
Let  $\overline{\beta},\overline{\gamma}$ be a critical point of $\mathcal{R}_1^{\alpha}$. Then
$$\frac{\partial\mathcal{R}_1^{\alpha}(\overline{\beta},\overline{\gamma})}{\partial\beta}=
    w\csc\beta \left(-y \cos\Theta +\tfrac{3}{4}d + \tfrac{1}{2}x\right)
    =0,$$
and thus $\cos\Theta = \frac{3}{4}\big(\frac{d}{y}+\frac{2x}{3y}\big) = \tfrac{3}{4}\left(\sin 2\alpha\csc\Delta + \frac{2}{3}\sin \gamma \csc\Delta \right).$
Substituting in  $\frac{2}{3}=\cos\Delta$ (by Lemma \ref{Critical c for k1}) gives
$$\cos\Theta = \tfrac{3}{4}\left(\sin 2\alpha+\cos\Delta\sin\gamma\right)\csc\Delta,$$
which simplifies to $\cos\Theta = \frac{3}{4} \cos \gamma.$
\end{proof}

\begin{lemma}
\label{proof of optimality first Alpern}
The critical points $\overline{\beta},\overline{\gamma}$ of Lemma~\ref{Alpern first step proof} are local minima of $\mathcal{R}_1^{\alpha}$.
\end{lemma}
\begin{proof}
Let  $\overline{\beta},\overline{\gamma}$ be a critical point of $\mathcal{R}_1^{\alpha}$. Now, taking the second derivative of $\mathcal{R}_1^{\alpha}(\overline{\beta},\overline{\gamma})$ with respect to $\gamma$ gives
$$\frac{\partial^2\mathcal{R}_1^{\alpha}(\overline{\beta},\overline{\gamma})}{\partial\gamma^2}
    = d(\tfrac{3}{4}\cot^2\Delta-\csc\Delta\cot\Delta+\tfrac{3}{4}\csc^2\Delta),$$
the right hand side of which simplifies to 
$2d\cot\Delta(-\tfrac{1}{2}\csc\Delta + \tfrac{3}{4}\cot\Delta) + \tfrac{3}{4}d.$ Observe that $-\tfrac{1}{2}\csc\Delta + \tfrac{3}{4}\cot\Delta=\frac{\partial\mathcal{R}_1^{\alpha}(\overline{\beta},\overline{\gamma})}{\partial\gamma} = 0$ at a critical point, and thus 
$$\frac{\partial^2\mathcal{R}_1^{\alpha}(\overline{\beta},\overline{\gamma})}{\partial \gamma^2}= \tfrac{3}{4}d.$$
Now,
$$\frac{\partial\mathcal{R}_1^{\alpha}(\overline{\beta},\overline{\gamma})}{\partial\beta} = w\left(-\cot\Theta+ \big(\tfrac{3}{4}d+\tfrac{1}{2}x\big)\csc\beta \right),$$
which simplifies to 
$$\frac{\partial\mathcal{R}_1^{\alpha}(\overline{\beta},\overline{\gamma})}{\partial\beta} = w\csc\beta(\mathcal{R}_1^{\alpha}\left (\overline{\beta},\overline{\gamma}) - \cos\beta \right)\label{final beta e1}.$$
Differentiating and rearranging once more obtains
$$\frac{\partial^2\mathcal{R}_1^{\alpha}(\overline{\beta},\overline{\gamma})}{\partial \beta ^2} = -2\cot\Theta \frac{\partial\mathcal{R}_1^{\alpha}(\overline{\beta},\overline{\gamma})}{\partial\beta} + w = w$$
Similarly,
$$\frac{\partial^2\mathcal{R}_1^{\alpha}(\overline{\beta},\overline{\gamma})}{\partial\beta d\gamma} = w\csc\beta \frac{\partial\mathcal{R}_1^{\alpha}(\overline{\beta},\overline{\gamma})}{\partial\gamma}=0.$$ 

Then the determinant of the Hessian of $\mathcal{R}_1^{\alpha}(\overline{\beta},\overline{\gamma})$ is equal to $\tfrac{3}{4}wd,$ which is positive; thus, $\mathcal{R}_1^{\alpha}(\overline{\beta},\overline{\gamma})$ is locally convex and attains a minimum at its unique critical point.
\end{proof}

\begin{lemma}
\label{lem: conditions beta gamma}
The critical points of $\ren{\infty}{\alpha}{\beta}{\gamma}:\reals^2\mapsto \reals$ are the solutions to the system 
\begin{align*}
\cos\Theta &= \tfrac{3}{4}\cos\gamma\\
\cos\Delta &= \tfrac{2}{3}\cos\beta,
\end{align*}
where $\Delta= \beta+\gamma$ and $\Theta=2\alpha+\gamma$. 
\end{lemma}

\begin{proof}
Taking the first derivative of $\mathcal{R}_1^{\alpha}(\beta,\gamma)$ with respect to $\beta$ gives
$$\frac{\partial\mathcal{R}_1^{\alpha}({\beta},{\gamma})}{\partial\beta}=
    \frac{1}{2(2-x)}\left(4 \frac{\partial w}{\partial\beta}+3\frac{\partial d}{\partial\beta}+2 \mathcal{R}_1^{\alpha}(\overline{\beta},\overline{\gamma}) \frac{\partial x}{\partial\beta}\right).$$
By substituting in the appropriate derivatives $\tfrac{\partial w}{\partial \beta} = -w\cot(\alpha+\beta),$ $\tfrac{\partial d}{\partial \beta} = w d \csc \beta,$ and $\tfrac{\partial x}{\partial \beta} = w x \csc \beta,$ 
we attain
$$\frac{\partial \mathcal{R}_1^{\alpha}(\beta,\gamma)}{\partial\beta}=
    \frac{w\csc(\beta)}{2(2-x)}\left(4 y \cos(\Theta) +3d+2 x \mathcal{R}_1^{\alpha}({\beta},{\gamma})\right).$$
This derivative is zero at ${\beta}$ when and only when
\begin{equation}
\mathcal{R}_1^{\alpha}({\beta},{\gamma}) = \frac{4 y \cos(\Theta) +3d}{2x}\label{dbR1}
\end{equation}

Similarly, as $\tfrac{\partial w}{\partial \gamma} =0$, $\tfrac{\partial d}{\partial \gamma} =-d \cot (\Delta),$ and $\tfrac{\partial x}{\partial \gamma} =d \csc (\Delta),$ then
$$\frac{\partial \mathcal{R}_1^{\alpha}({\beta},{\gamma})}{\partial\gamma}=
    \frac{d\csc(\Delta)}{2(2-x)}\left(-3\cos(\Delta)+2 \mathcal{R}_1^{\alpha}({\beta},{\gamma})\right).$$
At a critical point, this derivative is zero and thus
\begin{equation}
\mathcal{R}_1^{\alpha}({\beta},{\gamma})=\frac{3}{2}\cos(\Delta). \label{dcR1}
\end{equation}
Assume that both equations \eqref{dbR1} and \eqref{dcR1} hold; then, equating both formulas for $\mathcal{R}_1^{\alpha}({\beta},{\gamma})$ gives
$$\frac{4 y \cos(\Theta) +3d}{2x} = \frac{3}{2}\cos(\Delta).$$

As $x$ is nonzero, $y$ must also be nonzero and thus solving for $\cos(\Theta)$ gives
$$\cos(\Theta) =\frac{3}{4}\left(\frac{x\cos\Delta+d}{y} \right).$$
Substituting in formulas for $d,x,$ and $y$ gives

\begin{equation}
\cos(\Theta) = \frac{3}{4}\cos{\gamma}. \label{Theta in terms of gamma}
\end{equation}

Now, returning to equation~\eqref{dcR1}, substituting in a formula for $\mathcal{R}_1^{\alpha}({\beta},{\gamma})$ and solving for $w$ gives
$$w = \tfrac{3}{4}((2-x)\cos\Delta-d).$$
Dividing both sides by $y$, substituting in explicit formulas, and simplifying gives
$$ \frac{\sin\alpha}{\sin\beta}= \frac{3}{4}\left(\frac{2\sin\Theta\cos\Delta}{\sin\beta}-\cos\gamma \right),$$
which solves for
$$\sin\alpha =\tfrac{3}{2}\sin\Theta\cos\Delta-\left(\tfrac{3}{4}\cos(\gamma)\right)\sin\beta.$$
Using~\eqref{Theta in terms of gamma}, we obtain
$$\sin\alpha =\tfrac{3}{2}\sin\Theta\cos\Delta-\cos\Theta\sin\beta,$$
which with some minor trigonometric manipulation gives
$$\cos\Delta = \tfrac{2}{3}\cos\beta.$$
\end{proof}

\begin{lemma}
\label{lem: good bounds opt beta gamma}
Let $\overline\beta, \overline\gamma$ be as described in the statement of Theorem~\ref{thm: opt parameters for inftyRB}. Then, for all~$\alpha\in (0,3/4)$, we have that 
$$
0 \leq \overline\beta, \overline\gamma \leq \pi/2-\alpha.
$$
\end{lemma}
\begin{proof}
The lemma is established numerically by plotting $\overline\beta, \overline\gamma$ against $\pi/2-\alpha$, see Figure~\ref{fig: good bounds opt beta gamma}.

\begin{figure}[h!]
\begin{center}
 \includegraphics[width=7.5cm]{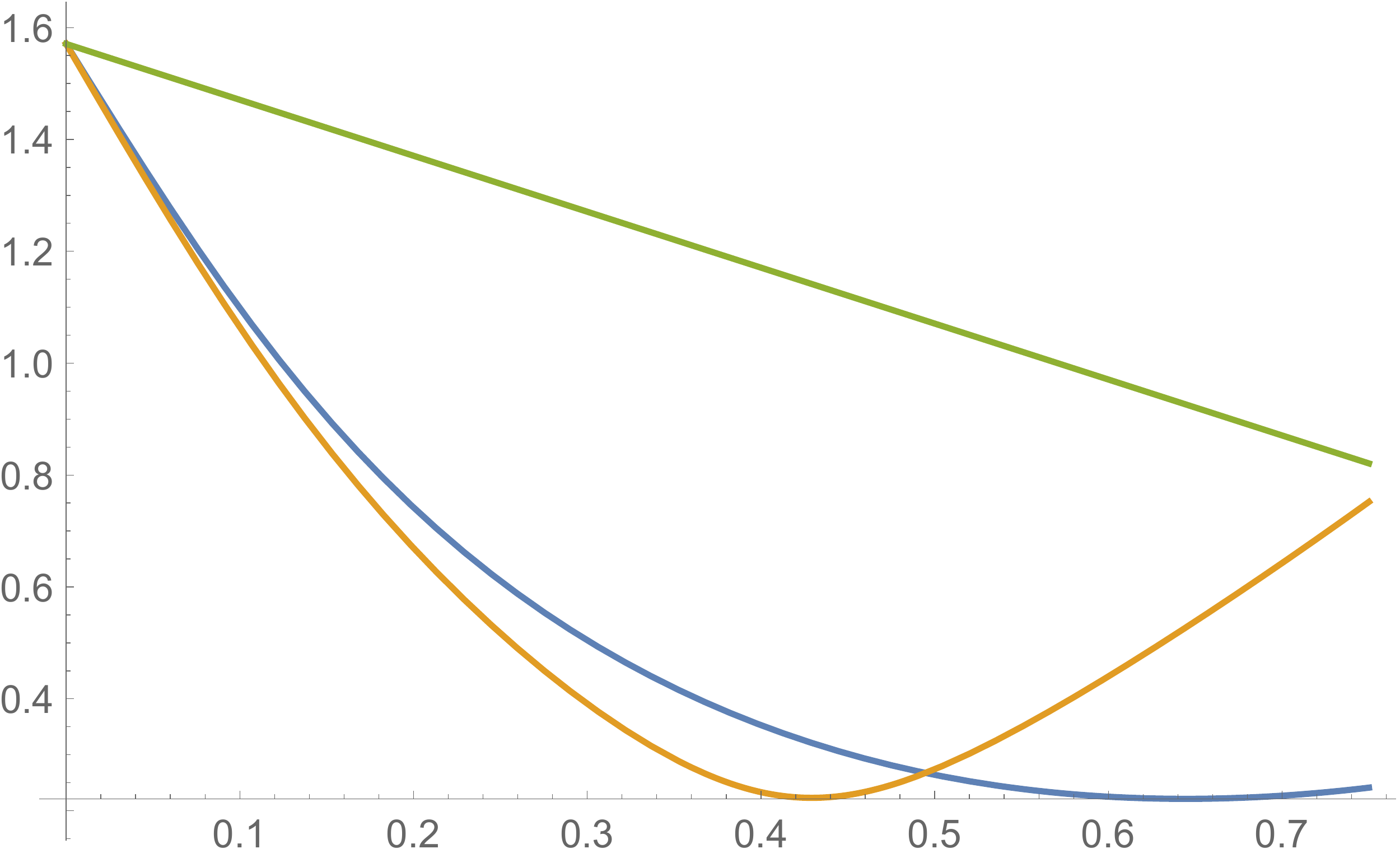}
\end{center}
 \caption{
Horizontal axis corresponds to $\alpha$. Graph depicts the behaviour of $\overline\beta=\overline\beta(\alpha)$ (blue curve) and of $\overline\gamma=\overline\gamma(\alpha)$ (yellow curve) as described in \eqref{equa: opt beta}, \eqref{equa: opt gamma}, respectively, of Theorem~\ref{thm: opt parameters for inftyRB}, 
against $\pi/2-\alpha$ (green line).  
 }
 \label{fig: good bounds opt beta gamma}
\end{figure}

\ignore{
w[\[Alpha]_, \[Beta]_] := Sin[\[Alpha]]*Csc[\[Alpha] + \[Beta]];
y[\[Alpha]_, \[Beta]_] := Sin[\[Beta]]*Csc[\[Alpha] + \[Beta]];
d[\[Alpha]_, \[Beta]_, \[Gamma]_] := 
  y[\[Alpha], \[Beta]]*Sin[2*\[Alpha]]*Csc[2*\[Alpha] + \[Gamma]];
x[\[Alpha]_, \[Beta]_, \[Gamma]_] := 
  y[\[Alpha], \[Beta]]*Sin[\[Gamma]]*Csc[2*\[Alpha] + \[Gamma]];
perf[\[Alpha]_, \[Beta]_, \[Gamma]_] := (3*
      d[\[Alpha], \[Beta], \[Gamma]] + 4*w[\[Alpha], \[Beta]])/
   2/(2 - x[\[Alpha], \[Beta], \[Gamma]])
va[a_] := (2*Cos[a] - Cos[2*a])*Csc[2*a]
barb[a_] := 
  ArcTan[(-va[a] + 
      Sqrt[va[a]^2 - (9/4*Cos[a]^2 - 1)*(5/4 - va[a]^2)])/(9/4*
       Cos[a]^2 - 1)];
barg[a_] := ArcCos[4/3*Cos[a + barb[a]]];
exprendtime[\[Alpha]_] := 
  perf[\[Alpha], barb[\[Alpha]], barg[\[Alpha]]];
(*compratio=*)
\[Rho]*exprendtime[ArcSin[1/\[Rho]]];
Plot[\[Rho]*exprendtime[ArcSin[1/\[Rho]]], {\[Rho], 1.5, 20}]
FindRoot[\[Rho]*exprendtime[ArcSin[1/\[Rho]]] == 425/100, {\[Rho], 10}]
Plot[ {barb[a], barg[a], Pi/2 - a}, {a, 0, 3/4}]
}
\end{proof}

\begin{lemma}
\label{lem: hessian of inftystep}
Let $\overline\beta, \overline\gamma$ be as described in the statement of Theorem~\ref{thm: opt parameters for inftyRB}. Then both eigenvalues of $\nabla^2\ren{\infty}{}{\overline\beta}{\overline\gamma}$ are strictly positive for all $0<\alpha<1/2$, and hence critical values $\overline\beta, \overline\gamma$ minimize $\ren{\infty}{}{\overline\beta}{\overline\gamma}$.
\end{lemma}
\begin{proof}
$\ren{\infty}{\alpha}{\beta}{\gamma}$ 
is given by~\eqref{equa: ren inf value} of Lemma~\ref{lem: krb performance}, so for all $\beta, \gamma$, we can compute $\nabla^2 \ren{\infty}{\alpha}{\beta}{\gamma}$. In the resulting $2\times 2$ matrix, we substitute the values 
$\overline\beta, \overline\gamma$, as in~\eqref{equa: opt beta},~\eqref{equa: opt beta} of Theorem~\ref{thm: opt parameters for inftyRB} to obtain 
$\nabla^2 \ren{\infty}{\alpha}{\overline\beta}{\overline\gamma}$ whose entries depends exclusively on $\alpha$. 
Using symbolic software, we calculate both eigenvalues of $\nabla^2 \ren{\infty}{\alpha}{\overline\beta}{\overline\gamma}$, and we verify that they are both strictly positive, for all $0< \alpha< 1/2$, see Figure~\ref{fig: hesianpd}.

\begin{figure}[h!]
\begin{center}
 \includegraphics[width=7.5cm]{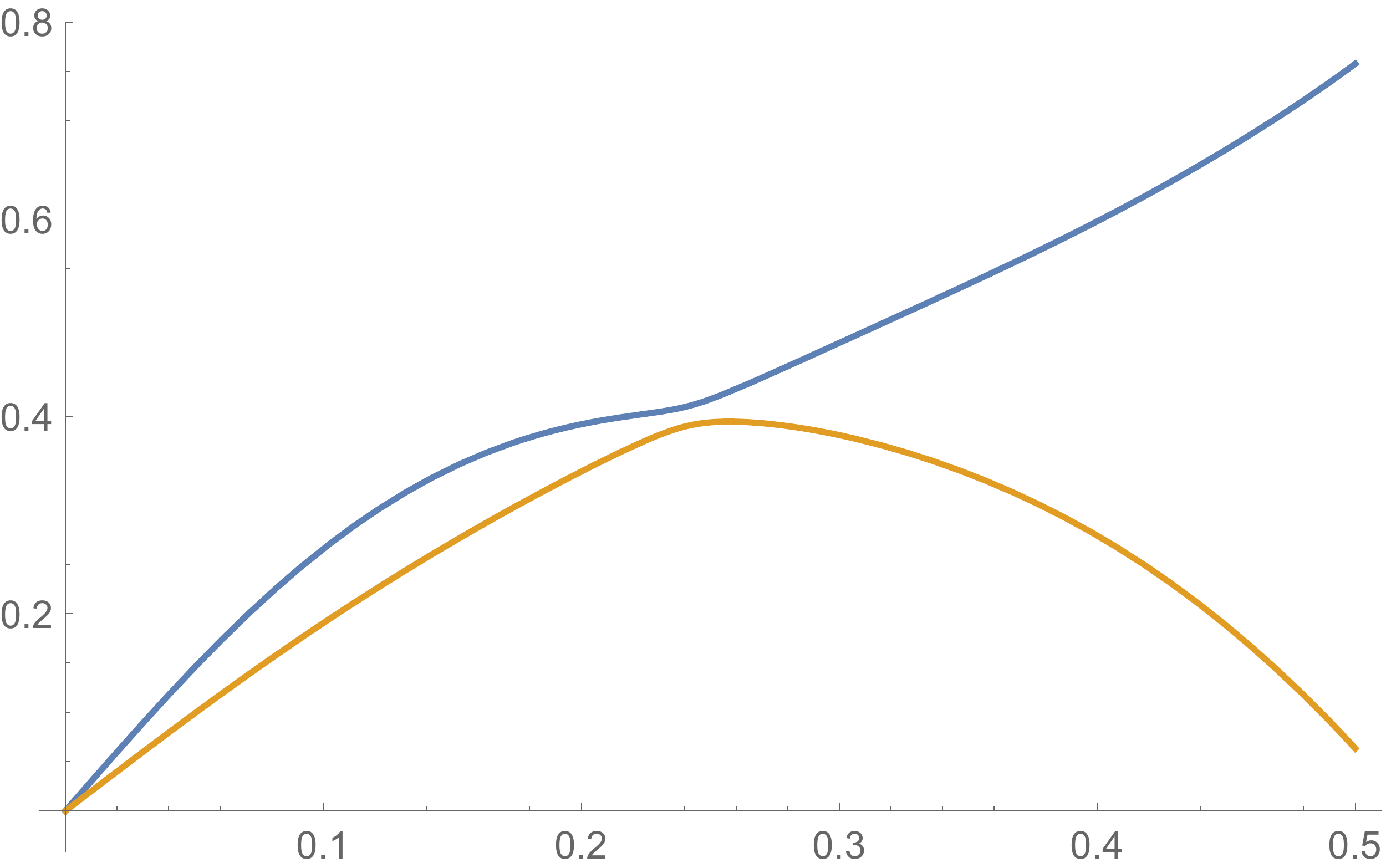}
\end{center}
 \caption{
The two eigenvalues of $\nabla^2 \ren{\infty}{\alpha}{\overline\beta}{\overline\gamma}$ as a function of $\alpha$. 
 }
 \label{fig: hesianpd}
\end{figure}
\end{proof}

\end{document}